\documentclass[aoas,preprint]{imsart}

\RequirePackage[OT1]{fontenc}
\RequirePackage{amsthm,amsmath,amssymb}
\RequirePackage{natbib}
\RequirePackage[colorlinks,citecolor=blue,urlcolor=blue]{hyperref}
\usepackage{subcaption} %
 \usepackage{epsfig}
 \usepackage{graphicx}	

\usepackage{algorithm} 
\usepackage{algpseudocode}
\usepackage{url}

\usepackage[table]{xcolor} 

\usepackage{verbatim} 

\usepackage[group-separator={,}]{siunitx} 

\usepackage{ctable} 
\usepackage{multirow}
\usepackage{booktabs,caption,fixltx2e}
\usepackage[flushleft]{threeparttable}
\arxiv{arXiv:1211.6807}

\startlocaldefs
\numberwithin{equation}{section}
\theoremstyle{plain}
\newtheorem{theorem}{Theorem}[section]
\newtheorem{lemma}[theorem]{Lemma} 
\newtheorem{proposition}[theorem]{Proposition} 
\newtheorem{corollary}[theorem]{Corollary}
\newtheorem{definition}[theorem]{Definition}
\endlocaldefs

\begin{document}

\begin{frontmatter}
\title{SCALABLE SPECTRAL ALGORITHMS FOR COMMUNITY 
DETECTION IN DIRECTED NETWORKS}
\runtitle{COMMUNITY 
DETECTION ALGORITHM IN DIRECTED NETWORKS}

\begin{aug}
\author{\fnms{Sungmin}
  \snm{Kim}\thanksref{t1}\ead[label=e1]{kim.2774@osu.edu}}
\and
\author{\fnms{Tao} \snm{Shi}\thanksref{t1}\ead[label=e2]{taoshi@stat.osu.edu}}

\thankstext{t1}{Partially supported by NSF grant (DMS-1007060 and DMS-130845)}
\runauthor{S. Kim and T. Shi}

\affiliation{Department of Statistics, The Ohio State University}

\address{1958 Neil Avenue\\
Columbus, OH, 43210-1247, USA\\
\printead{e1}\\
\phantom{E-mail:\ }\printead*{e2}}

\end{aug}

\begin{abstract}
  Community detection has been one of the central problems in network
  studies and directed network is particularly challenging due to
  asymmetry among its links. In this paper, we found that
  incorporating the direction of links reveals new perspectives on
  communities regarding to two different roles, source and terminal,
  that a node plays in each community.  Intriguingly, such communities
  appear to be connected with unique spectral property of the graph
  Laplacian of the adjacency matrix and we exploit this connection by
  using regularized SVD methods.  We propose harvesting algorithms,
  coupled with regularized SVDs, that are linearly scalable for
  efficient identification of communities in huge directed networks.
  The proposed algorithm shows great performance and scalability on benchmark
  networks in simulations and successfully recovers communities in
  real network applications.
\end{abstract}



\end{frontmatter}

\section{INTRODUCTION}
\label{sec:introduction}

Many real world problems can be effectively modeled as pairwise relationship 
in networks where nodes represent entities of interest and links
mimic the interactions or relationships between them. 
The study of networks, recently referred
to as {\it network science}, can provide insight into their structures and  
properties.
One particularly interesting problem in network studies is 
searching for important sub-networks which are called 
communities, modules or groups. 
A community in a network is typically characterized by a group 
of nodes that have more links connected within
the community than connected to other
nodes~\citep{Fortunato:2010p1496}.

In many practical applications, the networks
in study are directed in nature, such as the World Wide Web, tweeter's
follower-followee network, and citation networks. 
Compared with in-depth studies of community structures in undirected
networks~\citep{Danon:2005p2194, Fortunato:2010p1496,
  Coscia:2011p1986}, community detection in directed networks has not
been as fruitful.  We found that one particular reason is a
restrictive assumption on the community structure in a directed
network.  Many of previous studies assumed a member of a community has
balanced out-links and in-links connected to other members.
This symmetric assumption can be seriously violated in cases where a member may
only play one main role, source or terminal, in the community.

We give an example of such violation, Cora citation
network\footnote{\url{http://people.cs.umass.edu/~mccallum/data.html}
}, which presents bibliographic citations among papers in computer
science.  Citation networks have temporal restriction and it make them
difficult to form symmetric communities.
The left panal of Figure~\ref{fig:visual-intro} illustrates
communities detected by a popular Infomap algorithm~\citep{Rosvall:2009p2201}
which assumes symmetric communities.  We see that it only detects 
minuscule symmetric communities.  On the other hand, the right panel
of Figure~\ref{fig:visual-intro} reveals a completely different
community structure that turned out to show high correspondance to the
categories of the papers.  This result is obtained by relaxing the
symmetric assumption and allowing two different roles for a node in a
community. We defer our discussion on details of this example to Section~\ref{sec-6-1}.

\begin{figure}[htb]
\centering
\begin{subfigure}{0.35\textwidth}
\includegraphics[ width=1\linewidth]{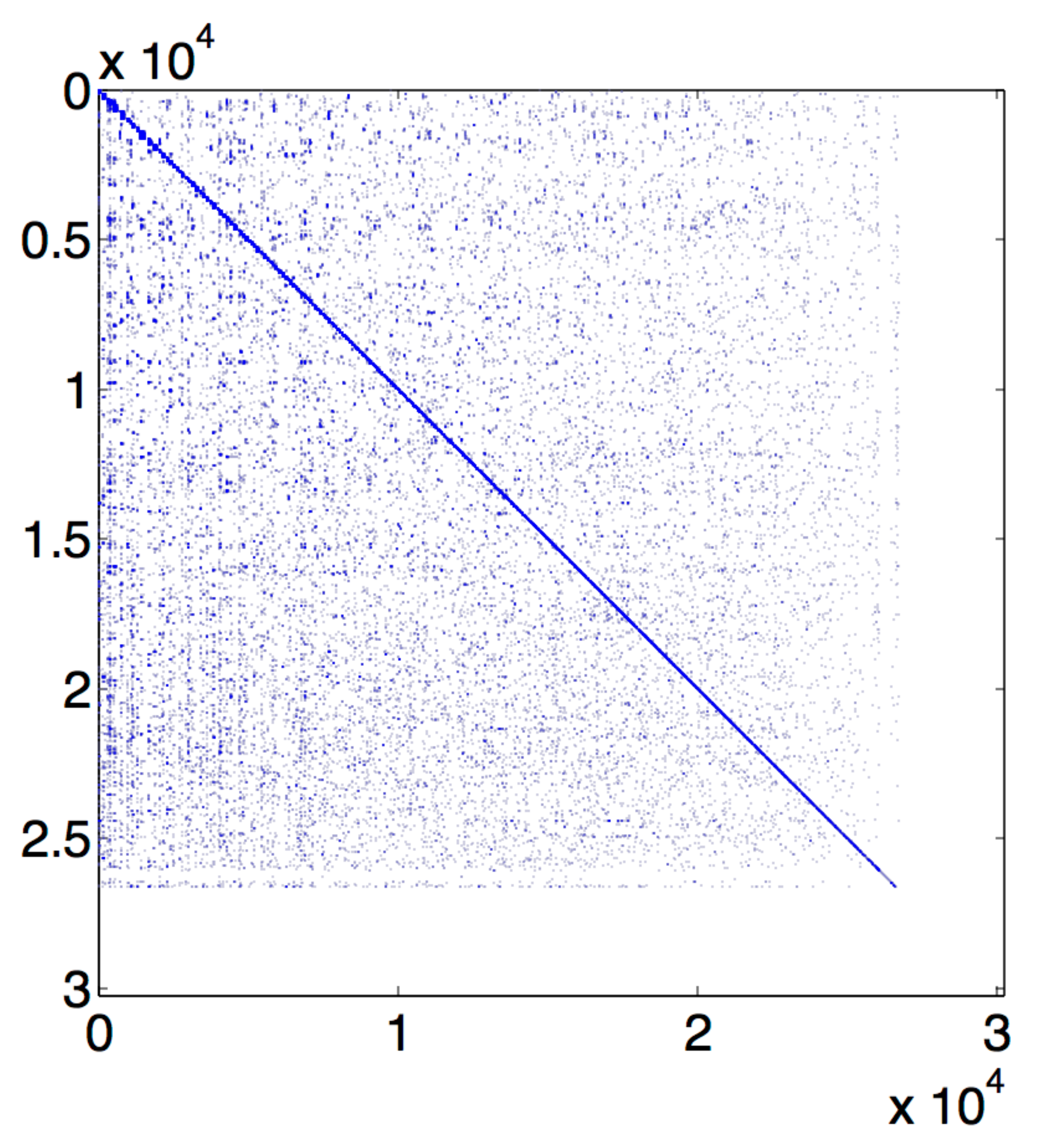}
  \caption{\label{fig:visual-info}Infomap }
\end{subfigure}%
\hspace{0.5in}
\begin{subfigure}{0.35\textwidth}
\includegraphics[ width=1\linewidth]{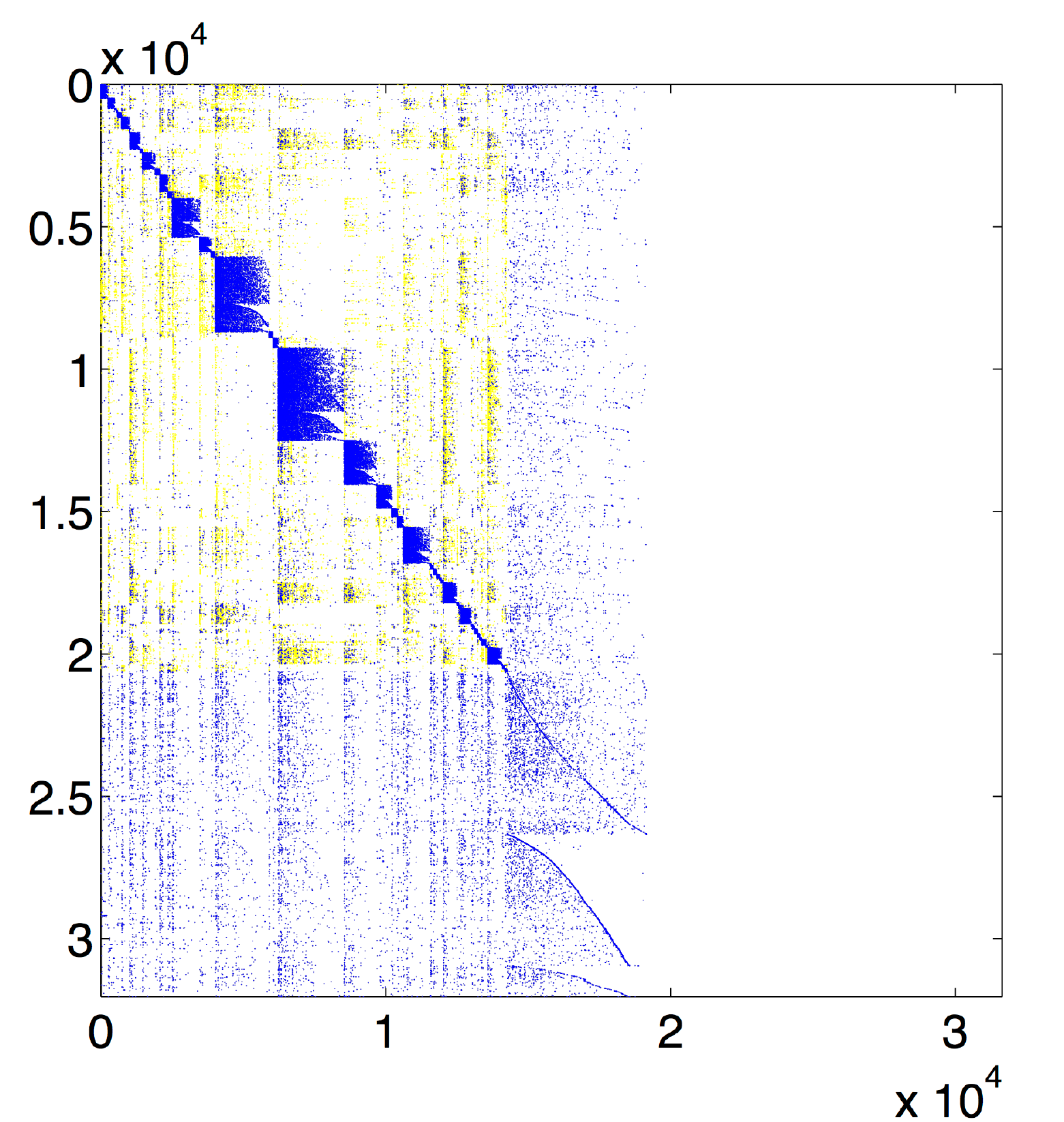}
  \caption{\label{fig:visual-L0}$L_0$-harvesting }
\end{subfigure}%
\caption{\label{fig:visual-intro} Comparison of symmetric communities
  detected by Infomap and asymmetric communities detected by the proposed algorithm
  in  Cora citation network. See
  Section~\ref{sec:comp-to-DI-SIM-and-Infomap} for more details.
}
\end{figure}

Asymmetric communities are common 
in directed networks where 
the direction implicitly express an asymmetric relationship among its nodes.
For example, social networks show celebrity-fan community structure and
celebrities hardly follow many fans.
 \cite{Satuluri:2011p1559} and
\cite{guimera2007module} also pointed out such asymmetric communities
in the context of World-Wide-Web, a Wikipedia network and investment
networks.

In this paper, we show that a community structure driven by 
two separate roles that a node plays in a directed network
can be formulated as a paired sets of
nodes.  We call such community, \textit{Directional Community}, which
is defined by a paired sets of nodes, a \textit{source node set} and a \textit{terminal
node set}.  We investigate notions of connectivity and quality measures
for a directional community. In those aspects, we propose algorithms
that is capable of detecting good directional communities.

Another aspect of a community detection algorithm we consider here is
scalability.  Huge networks raise two concerns,
computational complexity and computer memory requirements.  We exploit
regularized Singular Value Decomposition (SVD) and search local
communities in time proportional to the number of edges.

The remainder of this paper is organized as follows:
Section~\ref{sec:direct-netw-analys} introduces a new concept of
community for directed networks. Section~\ref{sec:regularizedsvd}
presents the regularized SVD algorithms developed for detecting the
communities. A simulation study is presented in Section~\ref{sec-5}. 
Section~\ref{sec-6} shows the results of the proposed algorithms in
two real-world networks. We finish with conclusions in
Section~\ref{sec-7}. 
Due to limitations of space, proofs of all theoretical results are
included in Appendix.

\section{COMMUNITY IN A DIRECTED NETWORK}
\label{sec:direct-netw-analys}

The nodes and links in a directed network are often presented by a
graph ${\cal G}=({\cal V},{\cal E})$, with ${\cal V} \equiv \{v_1,
\ldots, v_n\} $ and ${\cal E} \equiv \{e_{1} \ldots, e_{m}\}$ denoting
the vertex set and edge set respectively. For an existing edge $e$ in
the network, its source node and terminal node are denoted as $v^s(e)$
and $v^t(e)$ respectively.  
Let $W$ be the
$|{\cal V}| \times |{\cal V}|$ adjacency matrix in which $W(i,j) = 1$
indicates the existence of an edge originated from $v_i$ and pointed
to $v_j$ and $W(i,j) = 0$ otherwise.

In the literature of community detection in directed networks, several
authors have attempted to directly incorporate the directionality of edges
into their algorithms~\citep{capocci2005detecting,newman2007mixture,
  andersen2006communities,
  arenas2007size,Rosvall:2009p2201}.  
In particular, existing works pointed out the importance
of recognizing the dual roles, source and terminal of edges.
\citep{Zhou:2005,guimera2007module,benzi2012ranking}.

We consider a community structure that treats the dual roles
separately. 
A \textit{Directional Community} $C(S,T)$ is defined by two different sets
of nodes, a source node set $S \subset {\cal V}, S \neq \emptyset$, and
a terminal node set $T \subset {\cal V}, T \neq \emptyset$.  
Community structure is characterized by majority of edges placed within the
community (starting from the nodes in $S$ and ending at the nodes in
$T$).
In what follows, 
we first define a new type of connectivity between nodes in a directed
network. 
This newly defined connectivity 
leads to the concept of \textit{Directional Components},
which serve
as communities in the ideal situation in analogous to connected
components in an undirected network.
Furthermore, we consider a graph cut criterion that measures the quality
of a directional community.

\subsection{Directional Components}
\label{sec:directional-component}

We start with  
exploring connectivities of nodes in a directed network.
Two types of connectivity
have been studied in directed networks.
\emph{Weak connectivity} defines two 
nodes $s$ and $t$ ($s,t \in \cal{V}$) as weakly connected if 
they can reach each other through a path, 
regardless of the direction of edges in the path.
Meanwhile, \emph{Strong connectivity} follows the 
direction of edges in a path and calls nodes $s$ and $t$ 
strongly connected if the path \( (e_1, e_2, \ldots, e_l) \) 
also satisfies \(v^s(e_1) = s, v^t(e_l) = t,  v^t(e_k) = v^s(e_{k+1}),
k = 1, \dots, l-1 \). 
In this paper, we propose a new type of connectivity, 
\begin{definition}
Two nodes $s$ and $t$ ($s, t \in \cal{V})$ are {\textbf{D-connected}}, denoted by \( s \leadsto t \),
if there exists a path of edges
\( (e_1,  \dots, e_{2m-1}) \), $m \in \mathbb{R}^+$, satisfying $
v^s(e_1) = s, v^t(e_{2m-1}) = t $ and 
\begin{align}
\begin{cases}
 v^t(e_{2k-1}) = v^t(e_{2k}) &\mbox{(common terminal nodes)}\notag \\
 v^s(e_{2k}) = v^s(e_{2k+1}) &\mbox{(common source nodes)} \notag 
\end{cases}
\end{align}
for $k = 1, 2, \ldots, \max\{m-1,1\}$.
\end{definition}
D-connectivity follows the edges in alternating directions, first forward and then backward.
We call this sequence of edges a D-connected path.
Figure~\ref{fig:simple-example-1} provides an illustration of D-connectivity. 
For example, we observe that \( A \leadsto D \) through a sequence of edges \( (e_2, e_3, e_4) \) and 
 \( E \leadsto A \) through a sequence of edges \( (e_5, e_4, e_1) \). 

\begin{figure}[htb]
\centering
\includegraphics[width=0.75\textwidth]{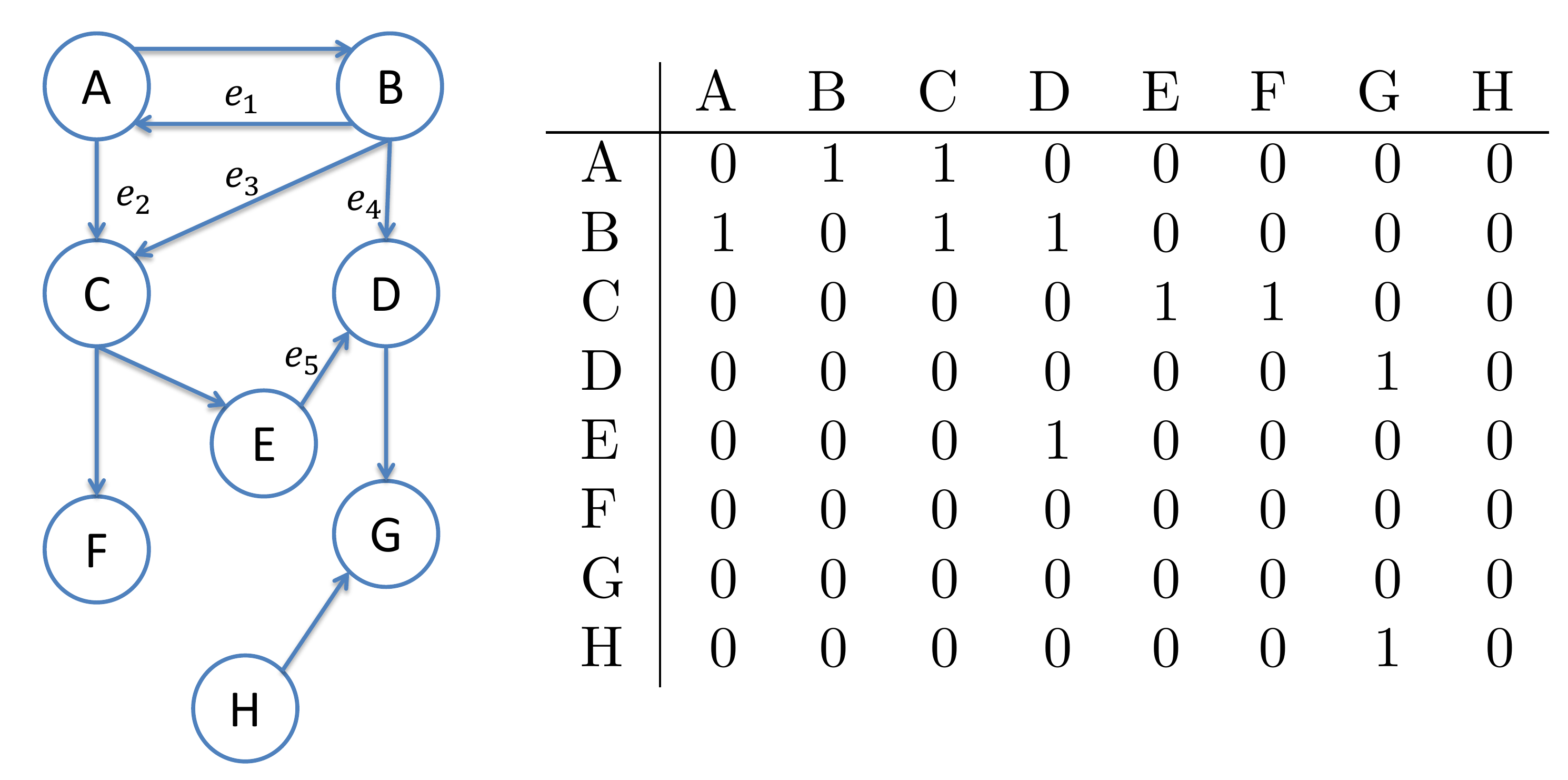}
\caption{\label{fig:simple-example-1} An example directed network and
  its adjacency matrix.}
\end{figure}

The definition of D-connectivity is a restricted version of a concept
called \emph{alternating connectivity} that was introduced by
\cite{Kleinberg:1999p1428} in the context of analyzing the centrality
of web-pages of World Wide Web using the HITS
algorithm.  The difference is that the alternating connectivity allows
two nodes be any pair on an alternating path regardless of
their roles.  Kleinberg also pointed out
the difficulty of developing the alternating connectivity to a concept
that characterizes a group of tightly connected nodes (a community),
because transitive relation does not hold in alternating connectivity.
However, D-connectivity bypasses this problem by recognizing the two
different roles, source and terminal.  Next we define a 
community structure, \emph{Directional Component}, based on the
D-connectivity.

\begin{definition} 
\label{def:dcomp}
A \textbf{Directional Component} $(DC)$ consists of a source node set
$S$ and a terminal node set $T$ $(S, T \subset \cal{V})$ and they are
the maximal subsets of nodes such that any pair of nodes \( (s,t) , s
\in S, t \in T \), are D-connected $( s \leadsto t )$. We call $S$
and $T$ the source part and terminal part of the directional component
and denote $DC \equiv (S, T)$.
\end{definition}
\noindent 
Directional components have desired properties as directional communities.
First, there is no edges between the source part of one component and
the terminal part of other components. 
Second, in a directed network that
contains multiple directional components $DC_1, DC_2, \ldots, DC_K$,
any node can belong to only one of the source
parts. In other words, the source parts $S_1, \ldots, S_K$ are
disjoint and the same holds for $T_1, \ldots, T_K$.  
Third, each edge belongs to one of directional components thus
they give a partition of edges. 

\begin{figure}[tb]
\centering
\includegraphics[width=.85\textwidth]{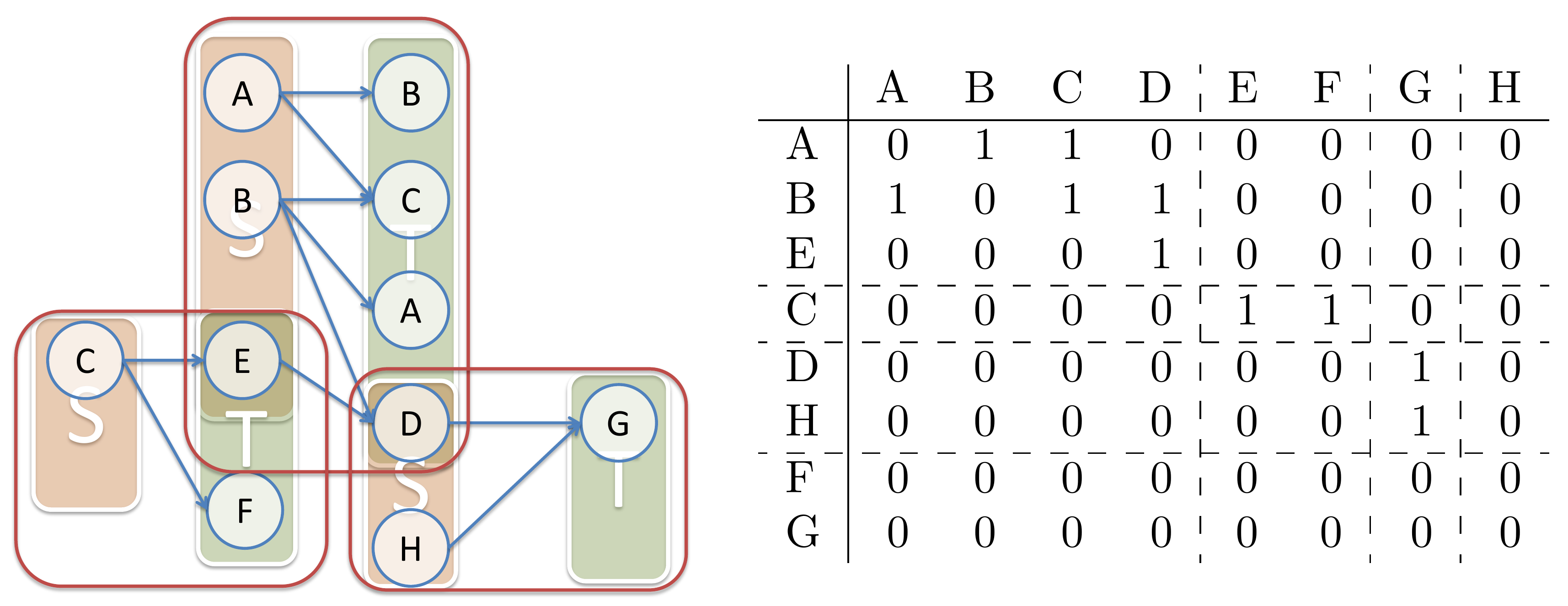}
\caption{\label{fig:simple-example-1d}The decomposition of 
the network in Figure~\ref{fig:simple-example-1} and rearranged adjacency
  matrix.}
\end{figure}

This two-way partition of nodes respects the asymmetric 
property of a directed network.
Figure~\ref{fig:simple-example-1d} illustrates the decomposition of 
the directed network shown in Figure~\ref{fig:simple-example-1}. 
Three directional components are found and the source and terminal parts of 
each directional component are displayed in boxes.
A node may have different memberships as source or terminal.
After reorganizing the nodes by directional components, 
there is no edges between the
source part and the terminal part of different directional components,
as shown in the right panel of Figure~\ref{fig:simple-example-1d}. 
This two-way partition of nodes results in a re-ordered 
adjacency matrix which exhibits block-wise structure.

A directional component may include the source part and the terminal
part that share few common nodes. This asymmetricity is possible
because the nodes in a D-connectivity path need to play only a single role,
source or terminal.
On the other hand, strong connectivity requires the nodes in a path,
except the first and the last node, to be source and terminal at the
same time.  Therefore, it is not surprising that many existing works
based on strong connectivity identify symmetric communities, for
example \cite{andersen2007local, Rosvall:2009p2201}.

Finding directional components 
can be achieved through a simple searching algorithm of computational
complexity \(O(|\cal{V}| + |\cal{E}|) \).  A directional component is
identified by iteratively adding nodes into the source part and the
terminal part (see Algorithm~\ref{al:directional-components} in 
Appendix~\ref{app-directional-components}).
We use this algorithm to decompose a directed network into directional
components prior to searching communities.

One drawback of searching for directional components is that real
networks usually have only one large directional component and
negligible small ones. This phenomenon is due to the fact that it is
unrealistic to expect absolutely no links between those communities.
In order to find more realistic communities, we first consider a
quality measure of directional community under the presence of a small
number of external edges.

\subsection{Directional Conductance}

We consider a graph cut criterion for
directional communities and define
directional cut between two directional communities $C_k(S_k,T_k),
C_l(S_l, T_l)$ as
\begin{equation}
  \label{eq:d-cut}
    \mbox{d-Cut}(C_k(S_k,T_k),C_l(S_l, T_l))  = \sum_{v_{i} \in S_k}
  \sum_{v_{j} \in T_l} W_{ij} + \sum_{v_{i} \in S_l} \sum_{v_{j}
  \in T_k} W_{ij}, 
\end{equation}
where $W$ is the adjacency matrix. 
Notice that two directional components have zero d-Cut. 

We want to emphasize the difference of d-Cut from a graph cut studied
in \cite{meila2007clustering}.  The graph cut criteria counts all
links between two communities while d-Cut only counts the links
starting from the source nodes of one community to the terminal nodes
of the other and vice versa.
Directional cut is equivalent to the graph cut criterion if $S_k=T_k, S_l = T_l$.

Based on d-Cut,  we propose a measure of the quality of a directional community,
\emph{Directional Conductance},
  \begin{equation}
       \label{eq:d-conductance}
       \phi(C(S,T))  = \frac{\mbox{d-Cut}(C(S,T), C(\bar{S},\bar{T})) }{\min\{\text{Vol}(S) +
         \text{Vol}(T),\text{Vol}(\bar{S}) +
         \text{Vol}(\bar{T}) \}}, \\
  \end{equation}
  where $\bar{S}, \bar{T}$ denotes the complement set of $S$ and
  the complement set of $T$, respectively.
$\mbox{Vol}(S) $ is defined as $\sum_{v_i \in S} d_{r,i}$, the sum of out-degrees of nodes in $S$ and
$\mbox{Vol}(T) $ is $\sum_{v_j \in T} d_{c,j}$, the sum of in-degrees of nodes in $T$.
The value of $\phi$ has a range from zero to one and a lower value
indicates relatively fewer external edges. 
Note that the value of $\phi$ for a directional component is zero. 
There are alternative normalizations that can be defined using d-Cut,
however, in this paper we concentrate on \eqref{eq:d-conductance}.

So far, we have explored the asymmetric roles of nodes in a directional
component.  The proposed D-connectivity preserves the roles along the alternating
paths and directional components divide a directed network into groups
according to the D-connectivity.  The distinction between the source part
and the terminal part of a directional community leads to the
directional conductance.
In the following sections, we develop scalable algorithms that
identify directional communities under the consideration of 
the connectivity and the conductance.

\section{REGULARIZED SVD ALGORITHMS FOR COMMUNITY EXTRACTION}
\label{sec:regularizedsvd}

\cite{KarlRohe:2012p1560} proposed a DI-SIM algorithm that uses the
low-rank approximation via SVD for bi-clustering (co-clustering) of
nodes in a directed network.
They investigated the spectrum of graph Laplacian of the adjacency matrix 
$W$. 
The graph Laplacian $Q$ is defined as  
\begin{equation} \label{def_laplacian}
Q =D_r^{-\frac{1}{2}} W D_c^{-\frac{1}{2}},
\end{equation}
where \(D_r \) is the diagonal matrix of out-degrees $\{d_{r,i}\}_{i =
  1, \ldots, n}$, and \(D_c\) is the diagonal matrix of in-degrees
$\{d_{c,j}\}_{j = 1, \ldots, n}$ \footnote{We define $\frac{0}{0} = 0$
  for convenience.}.  As a remark, the graph Laplacian $Q$ here is
different from the graph Laplacian considered in
\cite{chung2005laplacians,boley2011commute}, which is based on the
strong connectivity of nodes.  Assuming a known number of communities,
the DI-SIM algorithm clusters nodes in two different ways by running the k-means
algorithm on the leading left singular vectors and right singular
vectors separately.  They showed that the DI-SIM algorithm may recover
stochastically equivalent sender-nodes and receiver-nodes under
Stochastic Co-Block model, which is a relaxed version of Stochastic
Block model of \cite{holland1983stochastic}.

\subsection{Regularized SVD with $L_0$ Penalty}
\label{sec:regularized-svd-with}

In spite of DI-SIM algorithm's solid theoretical basis, there are
several limitations for our purpose on discovering directional
communities in a huge directed network.
First, it clusters nodes in two different clusters, but does not
provide paired source nodes and terminal nodes. 
Second, it requires a pre-specified number of communities, which is unknown
in most of applications.
Third, the spectral clustering is not scalable and not easy to be
parallelized.  
Huge networks frequently have many small
communities \citep{Leskovec:2008p2197} and it is challenging 
to recover all those communities at once.

In response to these limitations, we consider local-searching
algorithms to identify one community at a time rather than attempting
to discover all communities by the division of nodes. 
We propose a rank one regularized SVD by imposing $L_0$ penalty on
vectors $\mathbf{u}$ and $\mathbf{v}$ as follows, 
  \begin{equation}
      \label{eq:ideal}
      \max_{\mathbf{u} ,\mathbf{v} } \, \mathbf{u}^tQ \mathbf{v}  -\eta(\|\mathbf{u}\|_0 +
      \omega \|\mathbf{v}\|_0), \qquad \|\mathbf{u} \|_2 \leq 1,  \quad \|\mathbf{v}\|_2 \leq 1,
   \end{equation}
  where $\eta > 0$ is a penalty parameter and $\omega > 0$ determines
  the balance between the source part and terminal part. 
  \noindent The solution $(\mathbf{u}, \mathbf{v})$ from (\ref{eq:ideal}) leads to a community 
  $C(S, T)$ with $S = \{v: \mathbf{u}(v) \neq 0 \}$ and $T = \{v:
  \mathbf{v}(v) \neq 0\}$.

  Regularized SVD algorithms have been applied for bi-clustering
  tasks.  \cite{Lee:2010p1493,Witten:2009,Yang:2011p1492} showed how
  regularized SVDs cluster observations and features simultaneously.
  Results of bi-clusing typically show ``block-wise structure''.
  Such structure can be also found in the adjacency matrix of a
  directed network that has strong directional communities.

  We found this regularized SVD approach finds good directional
  communities in simulations and applications. The reasons are
  investigated in two perspectives. First, the regularized SVD problem
  is an approximation to minimizing directional
  conductance~\eqref{eq:d-conductance} with a penalty on the size of a
  community. Second, its solution leads to D-connected community.

\subsubsection{Approximately Minimize Penalized Directional Conductance}
\label{sec:regularized-svd-with-1}

Minimization of directional conductance over all possible directional
communities has two major limitations.  First, minimization of
conductance usually results in a balanced division of a
graph~\citep{Kannan:2004p1564} and recursive division of sub-graphs is
expensive for large networks.  Second, finding global minimization of
the criterion is NP-hard like the case in undirected networks.
Regarding the first limitation, we penalize the size of communities in
addition to the conductance. For the second limitation, we consider 
a spectral relaxation method to obtain
approximate solutions.

  First, we define the size of a community $C(S, T)$ as
  \begin{equation}
  \label{def:size}
  SZ_{\omega}(C(S,T)) \equiv |S|+ \omega |T|,
  \end{equation}
  where the constant $\omega>0$ balances the sizes of $S$ and $T$.
  Let us consider a quality measure of a directional community,
  \begin{equation}
       \label{eq:penalized}
       \phi_{\eta}(C(S,T)) = \frac{\mbox{d-Cut}(C(S,T), C(\bar{S},\bar{T}))}{\mbox{Vol}({S}) +
         \mbox{Vol}({T}) }  + 2 \eta {} SZ_{\omega}(C(S,T)),
  \end{equation}
  where $\eta > 0$ is a parameter determining the trade-off between
  conductance and the size of community.  $\phi_{\eta}$ penalizes
  large communities and prefers small communities having relatively
  low conductance.

  Now, we show that the regularized SVD problem \eqref{eq:ideal} is an
  approximation to the minimization of \eqref{eq:penalized}.  First, we
  introduce a proposition that reformulates $\phi_{\eta=0}(C(S,T))$.

  \begin{proposition}
  \label{prop:phi}
    Given a directional community $C(S,T)$, define two membership
    vectors
    $\mathbf{u}, \mathbf{v} \in \mathbb{R}^{n} $, 
  \begin{align}
     \label{eq:Ncut2}
     \mathbf{u}(v_{i}) & =
     \begin{cases}
     \frac{\sqrt{d_{r,i}}}{\sqrt{\text{Vol}(S) + \text{Vol}(T)}}, & v_i \in S \\
     0, & v_i \in \bar{S}, 
     \end{cases} 
      &
     \mathbf{v}(v_{j}) & =
     \begin{cases}
     \frac{\sqrt{d_{c,j}}}{\sqrt{\text{Vol}(S) + \text{Vol}(T)}}, & v_j \in T  \\
     0, & v_j \in \bar{T},
     \end{cases} 
  \end{align}
  then the following equations hold,
   $\phi_{\eta=0}(S,T)  = 1 - 2 \mathbf{u}^t Q \mathbf{v}$
and $\|\mathbf{u}\|_{2}^2 + \|\mathbf{v}\|_{2}^2 = 1$. 
  \end{proposition}
  This proposition can be proved by a standard result in
  graph cut theory that can be found in \cite{VonLuxburg:2007p47}. 
  The penalty on the size of community can be represented by
  $L_0$ penalty on $\mathbf{u}, \mathbf{v}$, since
  \begin{equation}
  SZ_{\omega}(C(S,T)) = \|\mathbf{u}\|_{0} + \omega \|\mathbf{v}\|_{0}.
  \end{equation}
Then, \eqref{eq:ideal} is obtained by the spectral relaxation that drops the discrete membership vector condition of
  $\mathbf{u},\mathbf{v}$ in \eqref{eq:Ncut2} and replacing $\|\mathbf{u}\|_{2}^2 +
  \|\mathbf{v}\|_{2}^2 = 1$ by $\|\mathbf{u}\|_{2} \leq 1, 
  \|\mathbf{v}\|_{2} \leq 1$. 

  Interestingly, we see that the penalty on the size of a
  community is actually a sparsity-inducing penalty on
  $\mathbf{u}, \mathbf{v}$.  Another interpretation of
  \eqref{eq:ideal} is a spectral relaxation of minimizing conductance
  with a sparsity inducing penalty.  It helps to
  recover the original sparse form of membership 
  vectors.

\subsubsection{Maintaining Directional Connectivity}
\label{sec:maint-direct-conn}

We introduced directional components in Definition~\ref{def:dcomp} and
showed they lead to block-wise structure of the adjacency matrix. 
For an undirected network, there is a well known relationship between
the spectrum of graph Laplacian and its connected
components~\citep{VonLuxburg:2007p47}: the
multiplicity of the largest eigenvalue (one) of Laplacian is the
same as the number of connected components in the network.  This
relationship can be extended to directed networks and directional
components.

For a subset of nodes, $A \subset {\cal V}$,  in a network with 
$n$ nodes,  we define $ \mathbf{1}_{A} $ 
as an indicator vector of length $n$ with each element $ \mathbf{1}_{A
  (i)} = I(v_i \in A) $. 
Recall that $S_k$ and $T_k$ represent the source part and terminal part of the 
$k$-th directional component respectively.
$Q(C(S,T))$ denotes a matrix obtained by replacing to zeros the rows and columns of $Q$ that are not in $S$ and $T$ respectively. We denote the principal singular value of a matrix $X$ by $\sigma_1(X)$. 

\begin{proposition}
\label{prop:d-comp} 
For a directed network, $\sigma_1(Q)$
is {\bf one} and its 
multiplicity, $K$, is equal to the number of directional components in the network.  
In addition, the principal left $($or right$)$ singular vector space is spanned by 
$ \{D_r^{\frac{1}{2}} \mathbf{1}_{S_{1}}, \ldots, D_r^{\frac{1}{2}} \mathbf{1}_{S_{K}} \}$ $($or 
$\{D_c^{\frac{1}{2}} \mathbf{1}_{T_{1}}, \ldots, D_c^{\frac{1}{2}}
\mathbf{1}_{T_{K}} \})$.
\end{proposition}
\noindent This proposition informs that a directional component is indeed a
solution of \eqref{eq:ideal} when $\eta = 0$. Moreover, when a network has only
one directional component, sufficiently large $\eta$ allows us to find
a D-connected subnetwork embedded in the directional component, as
stated in the following theorem.

\begin{theorem}
\label{thm:connectivity}
For any $\eta > 0$, the directional community derived from a solution of \eqref{eq:ideal} is a D-connected subgraph of a directional component.  Furthermore, the solution is a strict subgraph of a directional component if and only if the penalty parameter $\eta$ is greater than 
  \begin{equation}
   \min_{S,T}\frac{1-\sigma_{1}(Q(C(S,T)))}{SZ_{\omega}(DC_1) -
      SZ_{\omega}(C(S,T))}  \quad \text{subject to } SZ_{\omega}(C(S,T)) < SZ_{\omega}(DC_1),
  \end{equation}
where $DC_{1}$ denotes the smallest directional component. 
\end{theorem}
\noindent 
Combined with the relationship between the principal singular value and directional conductance in Section~\ref{sec:regularized-svd-with-1}, we expect the solution of \eqref{eq:ideal} to be not only D-connected but also to have low directional conductance relative to its size.

So far, we have discussed the properties of the directional community
 obtained by the $L_0$ regularized SVD formulation.
Next, we show that it can be solved efficiently by
using iterative matrix-vector multiplications combined with hard-thresholding.

\subsubsection{$L_0$ Regularized SVD Algorithm}

A local solution of \eqref{eq:ideal} can be found by iterative
hard-thresholding in the similar way of \cite{Shen:2008p1984} and
  \cite{dAspremont:2008p1869}. 
We start with exploiting the bi-linearity of the optimization problem \eqref{eq:ideal}.
For a fixed vector $\mathbf{v}$, we show how to solve the maximization problem 
with respect to $\mathbf{u}$. Here we first introduce a definition,

Given a vector $\mathbf{z}=(z_1, \ldots, z_n)' \in \mathbb{R}^n$, we denote the $l$-th largest 
absolute value of $\mathbf{z}$ as ${|z|}_{(l)}$. Consequently, we define $\mathbf{z}^h_l(\in \mathbb{R}^n)$ as the 
vector resulted from hard thresholding $\mathbf{z}$ by its
$(l+1)$-th largest absolute entry, i.e. the $i$-th element of
$\mathbf{z}^h_l$ is
\(\mathbf{z}^h_l(i) = z_i \, I(|z_i| > |z|_{(l+1)}) \)
while the superscript $``h"$ stands for hard-thresholding.

For a fixed $\mathbf{v}$, we may treat $Q \mathbf{v}$ as a generic vector $\mathbf{z}$ and find the 
solution $\mathbf{u}$ that maximizes \eqref{eq:ideal} through the following proposition.

\begin{proposition}
\label{prop:optim-L0}
For a given vector $\mathbf{z}$ and a fixed constant $\rho > 0$, the solution of 
\begin{equation} \label{def:maxL0}
\max_{\|\mathbf{u}\|_2 \leq 1} \mathbf{u}^t\mathbf{z} -\rho \|\mathbf{u}\|_0 
\end{equation}
is 
\[
\mathbf{u} = {\mathbf{z}^h_l} / {\|\mathbf{z}^h_l \|_2},
\]
where the integer $l$ is the minimum number that satisfies
\begin{equation}\label{eq:optim-L0}
|z|_{(l+1)} \leq \sqrt{ \rho^2 + 2 \,  \rho \, \| \mathbf{z}^h_l \|_2 }.
\end{equation}
\end{proposition}
\noindent
When the absolute values of $\mathbf{z}$ contains tied values, we pick one arbitrarily.

Proposition~\ref{prop:optim-L0} suggests a computationally efficient
algorithm to determine the threshold level. 
We
first sort the entries of $\mathbf{z}$ by their absolute values and
then sequentially search from the largest to smallest while testing if
condition \eqref{eq:optim-L0} has been met. As soon as
\eqref{eq:optim-L0} is satisfied, we obtain the hard-threshold level.
The computational complexity of this direct-searching algorithm is
$O(n\log (n))$.

Consequently, the solution of regularized SVD problem \eqref{eq:ideal} 
is obtained by using the searching algorithm
for a fixed $\mathbf{v}$ and for a fixed
$\mathbf{u}$ alternatively. Each step increases the objective function
monotonically, thus it converges to a local optimal.   
Algorithm~\ref{al:L0} lists the details.

\begin{algorithm}[t]
 \caption{$L_0$ regularized SVD}
\label{al:L0}
 \begin{algorithmic}[1]
   \Require {$ Q, \eta, \omega$}
\State initialize $\mathbf{v} $
\Repeat
\State $\mathbf{z} \gets Q\mathbf{v}$ , $\rho \gets \eta$
\State $\mathbf{u} \gets {\mathbf{z}^h_l} / {\|\mathbf{z}^h_l\|_2}$ , where
$l$ is the minimum integer s.t. $|z|_{(l+1)} \leq \sqrt{\rho^2  + 2 \, \rho \, \|\mathbf{z}^h_l\|_2}$
\State $ \mathbf{z} \gets Q^t\mathbf{u}$ , $\rho \gets \eta \, \omega$
\State $\mathbf{v} \gets {\mathbf{z}^h_l} / {\|\mathbf{z}^h_l\|_2}$, , where
$l$ is the minimum integer s.t. $|z|_{(l+1)} \leq \sqrt{\rho^2  + 2 \, \rho \, \|\mathbf{z}^h_l\|_2}$
\Until{ $ \mathbf{u}, \mathbf{v} $ converged}
\State \Return {$ \mathbf{u}, \mathbf{v} $}
 \end{algorithmic}
\end{algorithm}

The algorithm show similarity to HITS algorithm of
\cite{Kleinberg:1999p1428}, but algorithm~\ref{al:L0} uses Laplacian
matrix $Q$ instead of $W$.  Besides, the algorithm also has additional steps
that thresholds the membership vectors.  Consequentially, the algorithm can
detect a pair of sets of nodes constituting a local community instead it
converges to a principal singular vector of $Q$.

\subsection{Regularized SVD with Elastic-net Penalty}

In Section~\ref{sec:regularized-svd-with}, we find that the $L_0$
regularized SVD detects small and tight communities in direct networks
and it can be solved by an efficient algorithm based on the iterative
method combined with hard thresholding. 
Inspired by the discussion in \ref{prop:phi} about the
sparsity-inducing penalty, we also consider another type of penalty,
Elastic-net penalty of \cite{Zou:2005p1030} in a constraint form,
\begin{align}
\label{eq:objective-f-elastic}
& \max_{\mathbf{u} ,\mathbf{v} } \mathbf{u}^{t}Q \mathbf{v}, \\
 & \mbox{subject to } \quad (1-\alpha) \|\mathbf{u} \|_{2}^2 + \alpha \|\mathbf{u} \|_{1} \leq c_1,  
 \quad
 (1-\beta) \|\mathbf{v}\|_{2}^2 + \beta \|\mathbf{v}\|_{1} \leq c_2, \notag
\end{align}
where the sparsity level is controlled by parameters $\alpha \in [0,1)$ and $\beta \in [0,1)$. 
Note that $\alpha = \beta = 0 $ leads to the regular SVD problem. 
The optimization problem becomes non-convex when $\alpha \in (0,1)$ and $\beta \in (0,1)$.

We initially considered the constraint form of $L_0$ penalty in
order to search 
communities under a strict constraint on its size.  However, finding
a solution of the problem is challenging due to the discrete nature of
the constraint.  
We considered $L_1$ constraint
form that is proposed by \cite{Witten:2009}, but it did not report
significantly better solutions than $L_0$ regularized SVD solution 
\eqref{eq:ideal} in our simulation studies.  On the other hand,
the solution of Elastic-net constraint SVD shows different behavior 
than that of $L_0$ penalty.

\subsubsection{Elastic-net Regularized SVD Algorithm}
\label{sec:EN-algorithm}

We
show that a local solution of \eqref{eq:objective-f-elastic} 
can be found by the iterative method with soft-thresholding. 
Similar to the calculation of the $L_0$ regularized SVD, 
we take advantage of the bi-linearity of the optimization problem. 
For fixed $\mathbf{v}$ and $\alpha$, (or $\mathbf{u}$ and $\beta$), the optimization 
becomes convex, 
\begin{equation}
\label{eq:EN-u}
\max_{\mathbf{u} } \mathbf{u}^{t} \mathbf{z}, \quad
  \mbox{subject to } \quad (1-\alpha) \|\mathbf{u} \|^2_{2} + \alpha \|\mathbf{u} \|_{1} \leq c_1, 
\end{equation}
where $\mathbf{z} = Q\mathbf{v}$. 
Its global solution can be obtained through simple
soft-thresholding. 

We note that  \cite{Witten:2009}
showed similar results under slightly different
constraints. 
Our contribution is that we seek the soft-threshold level in the linear
time that is proportional to the number of non-zero entries of the
solutions, which makes the computation feasible for large matrix in
comparison to the binary search method proposed previously.
We verified that the linear search method is faster than the binary
search method by \(3\) to \(20\) times when the number of nodes in the network is
between \(10^3\) and \(10^7\).

To find the solution of \eqref{eq:EN-u}, we first introduce a
definition:
\begin{definition}
For a vector $\mathbf{z}=(z_1, \ldots, z_n)' \in \mathbb{R}^n$, recall the $l$-th largest 
absolute value of $\mathbf{z}$ was defined as ${|z|}_{(l)}$.
We define
\begin{align}
\label{eq:23}
G_{\mathbf{z}}(x) &=  \frac{1}{4x^2} \sum_{i = 1}^{k(x)-1} (|z|_{(i)}- x)^2+ 
\frac{1}{2x} \sum_{i = 1}^{k(x)-1} (|z|_{(i)} - x) 
\end{align}
where $ k(x) \in \{ 1, \dots , n+1 \} $ satisfies $ |z|_{(k(x))} \leq
x < |z|_{(k(x)-1)} $ and define $|z|_{(n+1)} = 0$. 
\end{definition}

We use a notation $S(\mathbf{z},d)$ for soft-thresholding of a vector
$\mathbf{z}$ by a scalar $d$. It is defined by
$S(\mathbf{z},d) = \operatorname{sign}(\mathbf{z}) (|\mathbf{z}| -
d)_{+}$, where $d > 0 $ and $x_{+} = \max\{x,0\}$. 

We find the solution
$\mathbf{u}$ that maximizes \eqref{eq:EN-u} by the following proposition:
\begin{proposition}
\label{thm:elastic}
For a fixed vector $\mathbf{z}$ and $\alpha$, the solution of \eqref{eq:EN-u} is 
\begin{align}
\mathbf{u} = {\frac{2d(1-\alpha)}{\alpha}} S(\mathbf{z}, d), \notag
\end{align}
and the threshold level $d$ is the solution of $G_{\mathbf{z}}(d) = c_1 ({1-\alpha})/{\alpha^2}$.
\end{proposition}

Then, a local solution of \eqref{eq:objective-f-elastic} can be found by
the alternative-maximization as in $L_0$ regularized SVD,
with steps 4 and 6 of Algorithm~\ref{al:L0} replaced by 
\begin{itemize}
\item [] step 4: $\mathbf{u} \gets  {\frac{2d(1-\alpha)}{\alpha}}
  S(\mathbf{z}, d)$ where $d$ satisfies $G_{\mathbf{z}}(d) = c_1 ({1-\alpha})/{\alpha^2}$
\item [] step 6: $\mathbf{v} \gets  {\frac{2d(1-\beta}{\beta}}
  S(\mathbf{z}, d)$ where $d$ satisfies $G_{\mathbf{z}}(d) = c_2 ({1-\beta})/{\beta^2}$.
\end{itemize}
The Algorithm involves searching for the soft threshold level $d$ in
equation $G_{\mathbf{z}}(d) = c$.
An efficient algorithm for finding the solution $d$ is
described in Algorithm~\ref{al:hatk} in Appendix~\ref{app-a}.

In summary, we propose two linearly scalable algorithms, the $L_0$
regularized SVD and the Elastic-net regularized SVD, for extracting
one community from a directed graph. In the next section, we will
apply these community extraction algorithms repeatedly to a network
and harvest tight communities sequentially.

\subsection{Community Extraction Algorithm}

We first emphasize the computational advantage of identifying one
community at a time for large networks. 
For example, \cite{clauset2005finding} discussed an approach of local
community detection in the application of World-Wide-Web, which cannot
even be loaded to a single machine's memory.
Algorithm~\ref{al:L0} uses only the
out-links of the current source nodes and the in-links of the current
terminal nodes in the matrix multiplication steps.  We will exploit
this property to devise a local community detection algorithm.

The regularized SVD algorithms require
the sparsity parameters, $\eta$ in \eqref{eq:ideal} or ($\alpha, \beta$) 
in \eqref{eq:objective-f-elastic} and a starting vector $\mathbf{v}$ or $\mathbf{u}$ to initialize the algorithm. 
In this section, we first discuss the effect of these 
parameters and how to choose them in practice. Then we propose a 
community harvesting scheme that repeatedly use the regularized SVD algorithm to 
extract communities.

\subsubsection{Parameter Selection and Initialization for Regularized SVDs}

We now study the effect of the 
penalization parameters on the algorithm outputs.
First, for Elastic-net regularized SVD, we point out that 
parameters $c_1$ and $c_2$  in \eqref{eq:EN-u} can be set to
one as default values, since they only affect the magnitude of the solution vectors.
Second, we find that imposing different sparsity to source nodes and
terminal nodes can be useful modification to the algorithms. 
However, we leave the investigation as a future work and we assume the
same sparsity levels in the rest of this paper.
Thus, we set $w = 1$ for $L_0$ regularized SVD and $\alpha = \beta$
for Elastic-net regularized SVD. 

We propose to use the directional conductance, $\phi(C(S,T))$ in
\eqref{eq:d-conductance} to find the best community among candidates.
Computing $\phi$ is inexpensive even for large networks
if degrees of nodes and the number of edges are known.  Although
$\phi$ may not be a good measure for comparing
communities in dramatically different sizes, it is still a decent measure for
similar-sized communities.  Thus, we will look for the community
achieving a local minimum value of $\phi$ over the smooth
change of communities.

Candiate communities are obtained by changing sparsity parameters 
($\eta$ for $L_0$ regularized SVD, $\alpha$ for $EN$ regularized SVD) 
smoothly. 
The solution $\mathbf{v}^{\ast}$ at the
current sparsity level can be used as the initial vector at the next
contiguous sparsity level.
The small change in the sparsity levels allows the algorithm converge
in few iterations without dramatic changes in solutions.
We consider a sequence of decreasing sparsity levels to obtain a
sequence of growing communities.  Starting with a small community,
this strategy lets us investigate relatively small communities in a
huge network by only visiting small fraction of the whole collection of edges.

We name the identified community $(S, T)$ from this method a
\emph{Approximated Directional Component} ($ADC$), to distinguish it
from directional components.  The steps of the algorithm are described
in Algorithm~\ref{al:c-d-comp}.  We note that one may simply replace the 
$L_0$ regularized SVD with Elastic-net regularized SVD.

\begin{algorithm}[tb]
 \caption{Community Extraction via $L_0$ Regularized SVD}
\label{al:c-d-comp}
\begin{algorithmic}[1]
\Require {$Q$, initialization vector
  $\mathbf{v}_0$, decreasing sequence of sparsity levels $\{\eta_i\}_{i = 1, \ldots, I}$}
\State initialize \(\mathbf{v} \gets \mathbf{v}_0\)
\For { \(i = 1\) to \(I\)}
\State Obtain \(\mathbf{u}^{*}, \mathbf{v}^{*}\) by running Algorithm~\ref{al:L0} with $(Q,
\eta_{i})$ and initialization \(\mathbf{v}\).
\State $S = \{v: \mathbf{u}^{\ast}(v) \neq 0 \}$ and $T = \{v: \mathbf{v}^{\ast}(v) \neq 0\}$
\State \(\phi_i \gets \phi(C(S,T))\) 
\State \((S^i, T^i) \gets (S,T)\),  \(\mathbf{v} \gets \mathbf{v}^{*}\)
\EndFor
\State \Return {$S = S^j$ and $T = T^j$,
 where $j$ corresponds to a local minimum in $\{\phi_1, \ldots, \phi_I\}$}.
\end{algorithmic}
\end{algorithm}

The algorithm requires a user to specify an initialization vector  $\mathbf{v}_0$ and a
sequence of the sparsity level parameters. 
The initialization vector $\mathbf{v}_0$ can be set as $\mathbf{1}_{ \{ v_{i}\} }$ with a randomly picked 
 $v_i$ with nonzero degree or can be set as the node with a large
 degree to discover the larger communities first. We use the later as
 default. 

 The searching for candidate communities can be stopped early if the
 conductance value reaches a local minimum of sufficiently low $\phi$.
 A simple implementation is to stop searching if the conductance value
 of the current candidate-$ADC$ bounces up to higher than $s_p$ ($s_p > 1$) times
 of the minimum conductance value of the previously detected
 candidate-$ADC$s.  Besides, we pre-specify a bound \(s_l\) $(0 < s_l <1)$ on the
 desired conductance value so we only stop searching early at a
 community with the conductance value lower than \(s_l\).  This
 stopping rule saves computation burden and keep the quality
 of communities.  We will use this early stopping rule in
 Section~\ref{sec-6}.

\subsubsection{Community Harvesting Algorithm}
\label{sec:comm-harv-algor}

In order to identify all tight communities in a directed network, we propose to apply
Algorithm~\ref{al:c-d-comp} repeatedly through a community harvesting scheme.
The idea of community extraction has been discussed in~\cite{zhao2011}, 
in which a modularity based method is introduced. 

Starting with the graph Laplacian matrix $Q$ of the full network, we
first apply Algorithm~\ref{al:c-d-comp} with $L_0$ or Elastic-net
penalty to identify an $ADC(S, T)$.  Then all entries in $Q$ that
corresponds to the submatrix of ($S, T$) are set to zero and we
reapply the same algorithm to the reduced \(Q\) matrix with a
different initialization to identify the next $ADC$.  It is continued
until the remaining edges are less than a pre-determined number $M$,
say $10\%$ of total number of edges.  Typically, the remaining
network contains tiny directional components which are originated from
the edges between communities.
We call this procedure \emph{harvesting} communities.

The harvesting algorithm takes different approach from other sparse
SVD algorithms for obtaining multiple sparse singular vectors.
\cite{Witten:2009} and \cite{Lee:2010p1493} use the residual matrix,
\( Q - s\mathbf{u}\mathbf{v}^t\) where \(s\) is pseudo singular value,
to obtain multiple sparse singular vectors. This approach does not fit
to our purpose because only the principal singular vector of a
submatrix is required for a directional component. In addition,
harvesting algorithm keeps the sparsity of \(Q\) whereas the other
approaches make the residual matrix more dense.

This scheme of harvesting edges of a detected community also allows
multiple memberships of nodes in both of source parts and terminal
parts.  On the other hand, this sequential removal of edges
may give a concern regarding the stability of the
detected communities. We observed the communities of low $\phi$ are
stable under different initializations and orders of
harvesting.

\subsection{Computational Complexity}
\label{sec:computation-times}
One driving motivation of this paper is the scalability of 
community detection algorithms on large or massive networks. 
Here, we investigate the harvesting algorithm's computational
complexity and memory requirement. 

In the specification of harvesting algorithms discussed in
Section~\ref{sec:comm-harv-algor}, there are four parameters: the
number of sparsity levels ($I$), the number of detected communities
($K$), the number of edges ($m$), and the number of nodes ($n$). The
complexity of a harvesting algorithm is \(O(IK(m+n\log{n}))\).  If the
optimal sparsity level is known, $I$ can be dropped.   Parallel
computing may potentially reduce the computation time by the factor of
$K$ if multiple communities can be searched simultaneously.

The computer memory requirement is mainly determined by $m$.
But for huge network data that cannot be fit into a machine, relatively
small sub-network can be explored locally. 
In this case, regularized SVDs only require a sub-network of $\sum_{v_i
  \in S} d_{r,i} + \sum_{v_i \in T} d_{c,i}$ edges and the source part
$S$ and the terminal part $T$ change smoothly over the iterations. 
We believe a parallel version of the harvesting algorithm is
a promising approach to tackle massive modern networks. 

The computation time may vary depending on the
settings of the algorithm and the data at hand. 
We report the computation time for the two 
large networks, a citation network and a social network, in Section~\ref{sec-6}.

\section{SIMULATION STUDY}
\label{sec-5}

In this section, we evaluate the performance of two harvesting
algorithms, $L_0$-harvesting and $EN$-harvesting under various
settings of community structures. 
 In addition to the harvesting algorithms, 
 the DI-SIM algorithm is included for comparison.

To generate networks with different types of community structures, we 
use a benchmark model proposed in \cite{LFR2008}, referred as 
LFR model. 
As a remark, currently the LFR model only generates symmetric directional
communities, which means the source part and the terminal part consist of
the same nodes. 
To generate a network with asymmetric communities, we shuffle the labels of
terminal nodes of the network generated from the LFR model. 

In our study, we generate networks from the LFR model with
$n=\num{1000}$ nodes, whose in-degrees follow a power law (with decay
rate $\tau_1 = -2$) with maximum at $k_{max} = 50$.  The sizes of the
communities in each network follow a power law with a decay rate
$\tau_2 = -1$ and the sizes of source part and terminal part are the
same.  We vary three sets of parameters of LFR model to control
different aspects of the simulated networks:

\begin{itemize}
\item[-] Range of community sizes: Set
  $(SZ_{\omega=1}(C)_{\mbox{min}}, SZ_{\omega=1}(C)_{\mbox{max}})$ as
$(40, 200)$ for big communities and $(20, 100)$ for small communities;
\item[-] Average degrees (in-degree and out-degree) \(k\): 
\(\{ 5, 10, 20\} \)
for sparse, median and dense networks;
\item[-] Proportion of external edges \(\mu\): 
\(\{0.05, 0.2, 0.4\} \). 
\end{itemize}

We measure the 
accuracy of community detection results 
by a mutual information based criterion that was proposed by
\cite{Lancichinetti:2009p2089}. 
The range of the accuracy measure is $[0 ,1]$ and the
larger the better.
Configurations of the algorithms in comparison are presented in
 Appendix~\ref{app-simulation-settings}.

The simulation results of thirty repetitions are reported in
Table~\ref{tab:simul}. 
We include \emph{Infomap} as a reference, 
 which shows the best
 performance in the report of  \cite{Lancichinetti:2009p1814}. 
As a remark, the accuracy of
\emph{Infomap} is measured in the symmetric directional communities
before shuffling the labels thus it is only valid as a
practical upper bound.  We emphasize that the performance of
Infomap on asymmetric directional communities is unsatisfactory in
general. 

\begin{table}[ht]
\caption{\label{tab:simul} 
Accuracies of four algorithms, $L_0$-harvesting,
 $EN$-harvesting, DI-SIM and Infomap in eighteen ($2 \times 3 \times
 3$) parameter combinations. The size
  of communities ranges in \(40 \sim 200\) for big communities and
  \(20 \sim 100\) for small communities. Average accuracies of thirty repetitions are reported along with standard errors. 
 The accuracies of Infomap
  cannot be directly compared to other methods since they are measured in symmetric directional
  communities while other three methods are applied on asymmetric
  directional communities.
}
\begin{center}
\begin{tabular*}{\textwidth}{ c|m{0.8cm}m{0.8cm}m{0.8cm}|m{0.8cm}m{0.8cm}m{0.8cm}|m{0.8cm}m{0.8cm}m{0.8cm}}
    \toprule
    \multicolumn{1}{c}{Degree}& \multicolumn{3}{c}{20} &
    \multicolumn{3}{c}{10} & \multicolumn{3}{c}{5}\\
$\mu$ & $0.05$ & $0.2$ & $0.4$ & $0.05$ & $0.2$ & $0.4$ & $0.05$ & $0.2$
& $0.4$\\
\midrule
\multicolumn{1}{c}{} &  \multicolumn{9}{c}{Big communities} \\
$L_0$ &$0.968$ \tiny{$(0.001)$}& $0.967$ \tiny{$(0.000)$}& $0.967$ \tiny{$(0.000)$} & $0.970$ \tiny{$(0.001)$}& $\bf{0.963}$ \tiny{$(0.001)$} & $\bf{0.778}$ \tiny{$(0.017)$}&  $\bf{0.924}$ \tiny{$(0.002)$} & $0.707$ \tiny{$(0.008)$}& $0.072$ \tiny{$(0.007)$}
 \\
$EN$ & $\bf{0.999}$ \tiny{$(0.003)$} & $\bf{0.999}$ \tiny{$(0.000)$}
& $\bf{0.994}$ \tiny{$(0.001)$} & $\bf{0.994}$ \tiny{$(0.001)$} & $0.956$ \tiny{$(0.002)$}& $0.195$ \tiny{$(0.011)$}& $0.851$ \tiny{$(0.008)$} & $0.446$ \tiny{$(0.017)$}& $0.023$ \tiny{$(0.004)$}\\
DI-SIM & $0.827$ \tiny{$(0.007)$}& $0.904$ \tiny{$(0.008)$}
& $0.946$ \tiny{$(0.007)$} & $0.840$ \tiny{$(0.007)$}& $0.910$ \tiny{$(0.006)$}& $0.237$ \tiny{$(0.006)$}& $0.845$ \tiny{$(0.007)$}& $\bf{0.786}$ \tiny{$(0.006)$} & $\bf{0.237}$ \tiny{$(0.011)$}\\
Infomap & $1.000$ \tiny{$(0.000)$}& $1.000$ \tiny{$(0.000)$}
& $1.000$ \tiny{$(0.000)$}  & $0.998$ \tiny{$(0.000)$}& $0.996$ \tiny{$(0.000)$} & $0.976$ \tiny{$(0.002)$}& $0.879$ \tiny{$(0.005)$}& $0.749$ \tiny{$(0.008)$} & $0.298$ \tiny{$(0.014)$}\\
\midrule
\multicolumn{1}{c}{} &   \multicolumn{9}{c}{Small communities} \\
$L_0$ & $0.953$ \tiny{$(0.002)$} & $0.917$ \tiny{$(0.010)$}
& $0.940$ \tiny{$(0.005)$} & $0.952$ \tiny{$(0.002)$} & $0.921$ \tiny{$(0.006)$} & $0.857$ \tiny{$(0.009)$}& $\bf{0.910}$ \tiny{$(0.003)$} & $0.696$ \tiny{$(0.009)$} & $0.106$ \tiny{$(0.007)$} \\
$EN$ & $\bf{0.996}$ \tiny{$(0.001)$} & $\bf{ 0.991}$ \tiny{$(0.001)$}
& $\bf{0.992}$ \tiny{$(0.001)$}& $\bf{0.984}$ \tiny{$(0.001)$}& $\bf{0.937}$ \tiny{$(0.003)$}& $0.526$ \tiny{$(0.019)$}& $0.883$ \tiny{$(0.004)$} & $0.585$ \tiny{$(0.007)$} & $0.061$ \tiny{$(0.004)$} \\
DI-SIM  & $0.762$ \tiny{$(0.005)$}& $0.821$ \tiny{$(0.006)$}
& $0.871$ \tiny{$(0.004)$}  & $0.762$ \tiny{$(0.005)$}& $0.847$ \tiny{$(0.005)$} & $\bf{0.894}$ \tiny{$(0.005)$}& $0.751$ \tiny{$(0.005)$} & $\bf{0.702}$ \tiny{$(0.005)$} & $\bf{0.309}$ \tiny{$(0.011)$}\\
Infomap & $0.999$ \tiny{$(0.000)$}& $0.997$ \tiny{$(0.001)$}
& $0.998$ \tiny{$(0.001)$} & $0.994$ \tiny{$(0.001)$}& $0.992$ \tiny{$(0.001)$}& $0.983$ \tiny{$(0.001)$}& $0.910$ \tiny{$(0.003)$}& $0.778$ \tiny{$(0.005)$}& $0.464$ \tiny{$(0.009)$}\\
    \bottomrule
\end{tabular*}
\end{center}
\end{table}

In the setting for big communities, Table~\ref{tab:simul} show
that our harvesting algorithms give almost perfect
recovery when average degree is high and/or mixing parameter $\mu$ is
low. 
$EN$-harvesting shows better accuracy than
$L_0$-harvesting for strong community settings while
$L_0$-harvesting excels in the low degree setting. 
The DI-SIM algorithm fails to give
high accuracies even for the strong communities. 

The accuracies of the harvesting algorithms remain close to the
results of big communities. 
However, DI-SIM algorithm seems to be less accurate
in the case of 
the larger number of communities.

 In our experiment, we also find that the performance of harvesting
 algorithms for strong communities is  
 close to that of \emph{Infomap} applied for symmetric communities. 
 \emph{Infomap} cannot detect asymmetric directional
 communities in contrast to harvesting algorithms.

\section{Applications}
\label{sec-6}
In this section, we apply the proposed harvesting algorithms to 
two highly asymmetric directed networks,
a paper citation network and a social network.

\subsection{Cora Citation Network}
\label{sec-6-1}
Cora citation network is
a directed network formed by citations among Computer Science (CS) research papers.
We use a subset of the papers that have been
manually assigned into categories that represent
10 major fields in computer science, which is further divided into 70
sub-fields. 
This leads to a network of  $\num{30228}$ nodes and 
$\num{110654}$ edges after removing self-edges.
In this citation network, only $5.4 \%$ of edges are symmetric. 
The average degree is \(3.66\), which is relatively low.

The algorithms start at the terminal nodes with
the largest in-degree among un-harvested nodes at each harvesting run.
The sparsity levels are determined so that candidate $ADC$s may cover
up to $50\%$ of nodes. 
The sparsity parameter $\eta$ in the \(L_0\)-harvesting takes
values decreasingly in a grid \( \{ \exp(-k): k = 10+i (8/200), i = 1,
\ldots, 200 \} \).  Similarly, the sparsity parameter $\alpha$
in the \(EN\)-harvesting takes values decreasingly in a grid \(
\{\frac{1}{1+\exp(k)}: k = 2 + i (7/200), i = 1, \ldots, 200 \} \).
The nonlinear decreasing setup helps to obtain gradual expansion of the
candidate-$ADC$s at low sparsity levels.  Early stopping parameters
are set to \(s_p= 1.4\) and \( s_l=0.4\).  Each algorithm runs
until it harvests $90\%$ of edges. 
$L_0$-harvesting discovered $51$
communities  in $4$ minutes and $EN$-harvesting discovered $78$
communities in $9$ minutes.

For both harvesting algorithms, we first provide a summary of the largest
twenty $ADC$s discovered. 
We name the $ADC$ obtained in the \(L_0\)-harvesting 
\(ADC^{L_0}\) and the ones obtained by the \(EN\)-harvesting
\(ADC^{EN}\). 
Out of total \(\num{110654}\) edges, the first twenty \(ADC^{L_0}\)s
cover \(\num{82372}\) 
edges ($74\%$) and the first twenty \(ADC^{EN}\)s cover
\(\num{88756}\) edges ($80\%$).

Most of detected communities have larger source part than terminal part
which reflects the presence of late papers that are not yet cited much.
Overall, we found \(ADC^{L_0}\)s are better than \(ADC^{EN}\)s
based on the comparison of conductance values. 
This result is consistent with the simulations in
Section~\ref{sec-5} that \(L_0\)-harvesting performs better 
in networks of low average-degrees  
(For more details, see Table~\ref{table:summary-d-comps} in Appendix~\ref{app-c}).

\subsubsection{Comparison to DI-SIM and Infomap}
\label{sec:comp-to-DI-SIM-and-Infomap}
To highlight the existence of asymmetric directional communities, two
existing community detection algorithms are also applied for
comparisons.  First, the DI-SIM algorithm \citep{KarlRohe:2012p1560}
is applied as an example of methods providing two separate
partitions. The required number of communities is set as the number of
major-fields in CS, which is ten.
Second, we applied the \emph{Infomap} algorithm of
\cite{Rosvall:2009p2201} with details given in Appendix~\ref{app-simulation-settings}.
We provide a visual comparison of communities detected by those four
algorithms in Figure~\ref{fig:visual-intro} and \ref{fig:visual}.

\begin{figure}[htb]
\centering
\begin{subfigure}{0.35\textwidth}
\includegraphics[ width=1\linewidth]{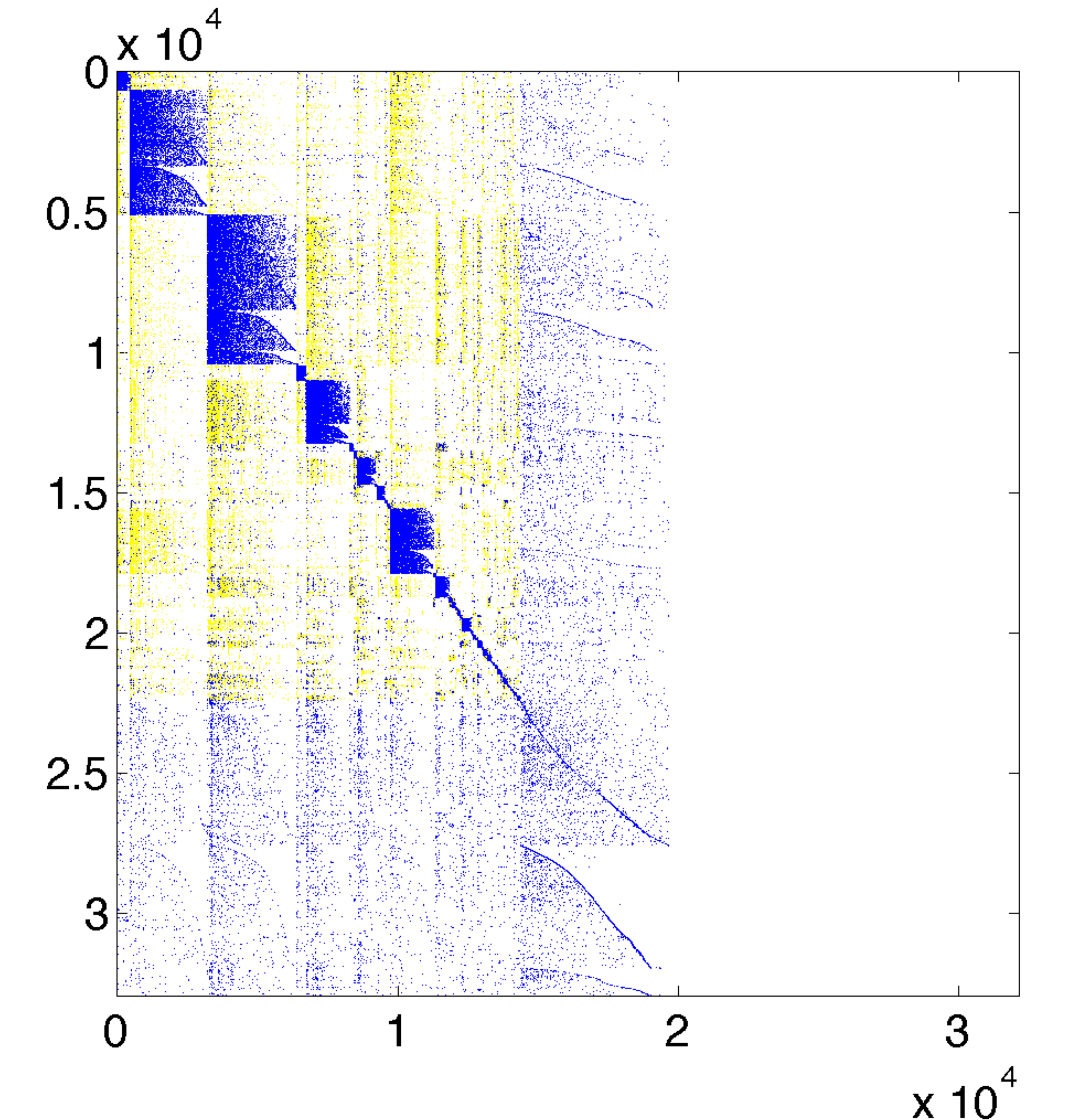}
  \caption{$EN$-harvesting }
\end{subfigure}%
\hspace{0.5in}
\begin{subfigure}{0.35\textwidth}
\includegraphics[ width=1\linewidth]{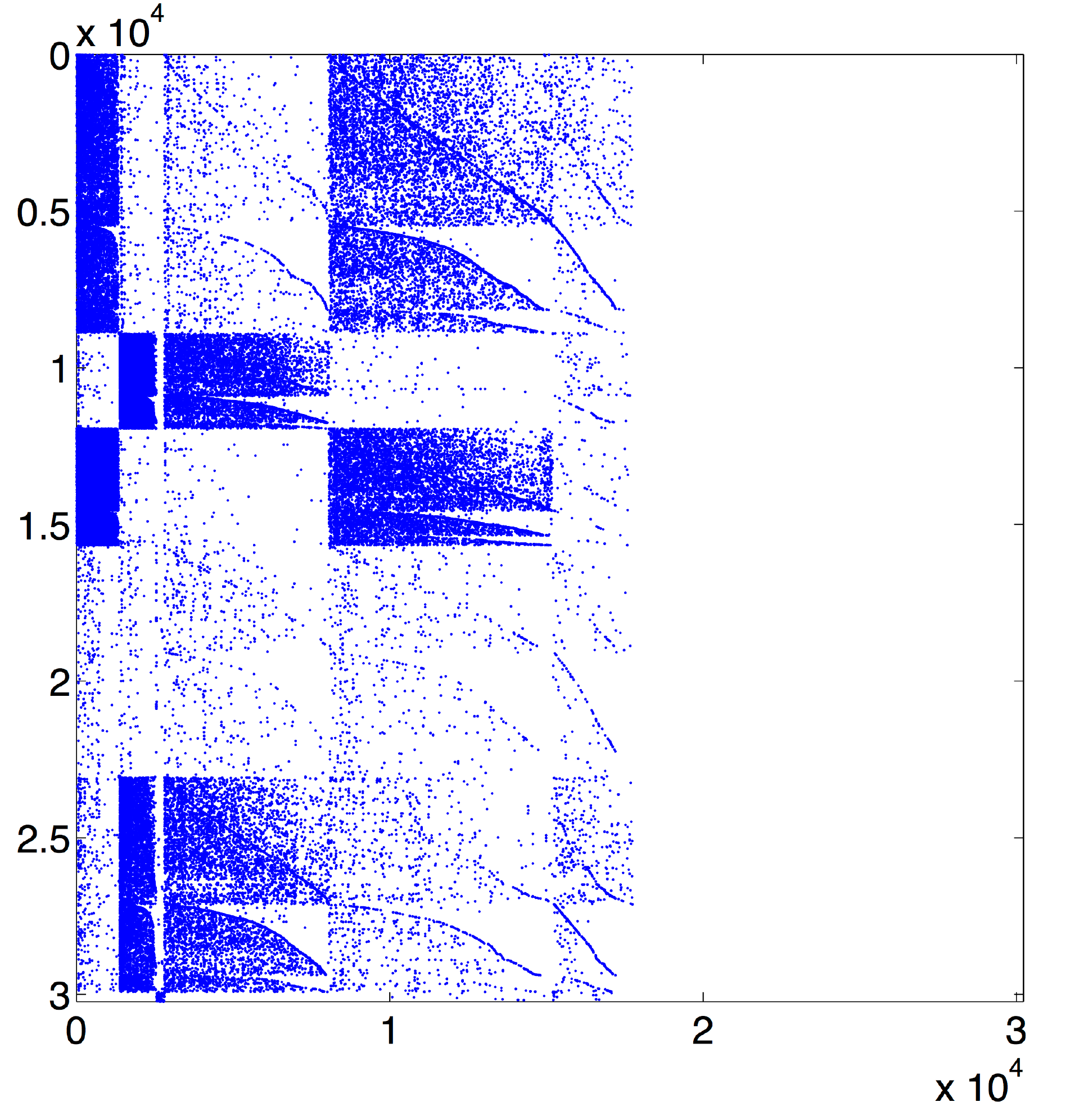}
  \caption{ \label{fig:visual-DI}DI-SIM }
\end{subfigure}%
 \caption{\label{fig:visual} (a): Cora citation network arranged by directional communities.  The rows and columns are arranged by the source parts
   and the terminal
 parts of the 78 $ADC^{EN}$s and remaining
 nodes are appended at the end of rows and columns. 
 (b): Cora citation network with rows and columns reordered by the
 result of the DI-SIM algorithm.  }
\end{figure}

Visualization of the communities detected by harvesting algorithms
is not straightforward due to the possibility of multiple memberships.
The rows and columns are arranged by the source parts and the terminal
parts of the detected $ADC$s and the remaining
nodes are appended at the end of rows and columns. 
 Internal edges of $ADC$
appear as blue blocks in the diagonal.
Meanwhile, blue dots outside the
blocks are the edges that are not harvested in the algorithm.
Yellow dots are the internal edges that reappear because of the
multiple membership of nodes.

The result from the DI-SIM algorithm is summarized by the
adjacency matrix with rows and columns reordered by the
two separate partitions (Figure~\ref{fig:visual-DI}).
The adjacency matrix rearranged by the communities of
\emph{Infomap} is shown in Figure~\ref{fig:visual-info}.

The communities detected by harvesting algorithms reveal different
representation of the underlying structure than other two methods.
First, harvesting algorithms capture asymmetric nature of the
communities in the citation network. The symmetric assumption of
Infomap yields tiny communities that are less significant.
Second, the proposed algorithms reveal correspondence between source
nodes and terminal nodes while DI-SIM treats them separately.

\subsubsection{Correspondence between Communities  and Manual Categories}
\label{sec:6-manual}

The manually assigned categories of papers (Table~\ref{table:fields}) in Cora citation network 
provided us with extra information to validate the
quality of the detected communities.
The sizes of category span a large range, from $582$ papers in
Information Retrieval to $\num{10784}$ papers in Artificial
Intelligence.
Given the categories, we calculate the conductance value of each category
to see the quality of a category as a community. 
Those values are greater than those of \(ADC^{L_0}\)s in
general (See Table~\ref{table:summary-d-comps} in Appendix~\ref{app-c}).

\begin{table}[ht] \footnotesize
\caption{\label{table:fields} List of ten fields of Computer Science
  and their number of papers and conductance.}
\centering
\begin {tabular}{r|l|r|r}
\toprule
Number & Name of Major Field of CS                   & Number of Papers & $\phi$ \\ 
\midrule
1  & Artificial Intelligence               & 10784 &  0.1568 \\   
2  & Data Structures Algorithms and Theory & 3104  &  0.3854 \\
3  & Databases                             & 1261  &  0.3429 \\
4  & Encryption and Compression            & 1181  &  0.4096 \\
5  & Hardware and Architecture             & 1207  &  0.4762 \\      
6  & Human Computer Interaction            & 1651  &  0.4527 \\
7  & Information Retrieval                 & 582   &  0.3932 \\      
8  & Networking                            & 1561  &  0.3686 \\
9  & Operating Systems                     & 2580  &  0.3736 \\      
10 & Programming                           & 3972  &  0.3178 \\  
 \bottomrule                                
\end{tabular}
\end{table}

\begin{table}[ht]
\caption{\label{table:cate-L0}Number of papers in the first six 
  approximated directional components of \(L_0\)-harvesting for each category. }
  \centering
\begin{tabular}{l|rrrrrrrrrrr}
  \toprule
     &    AI  &  DSAT  &   DB  &   EC  &   HA  &  HCI  &  IR  &  Net  &   OS  &  Prog  \\
  \midrule
1    &       106  &   199  &   56  &   18  &  255  &   17  &   0  &   55  &  900 \cellcolor[gray]{0.9} &  1779 \cellcolor[gray]{0.9}  \\
2     &      2741 \cellcolor[gray]{0.9}  &    68  &   30  &    9  &    8  &   28  &  63  &   17  &   34  &    75  \\
3     &      13  &    12  &   25  &  115  &   11  &  307  &  18  &  936 \cellcolor[gray]{0.9}  &  232  &    34  \\
4     &      727  \cellcolor[gray]{0.9} &   124  &    8  &  102  &   12  &  577  \cellcolor[gray]{0.9} &  10  &    6  &   22  &    21   \\
5     &      284  &    83  &  803 \cellcolor[gray]{0.9}  &    9  &    9  &   17  &  66  &   14  &   16  &    80   \\
6     &      149  &   452  \cellcolor[gray]{0.9}&    3  &  239 \cellcolor[gray]{0.9} &   12  &    0  &   2  &    3  &    9  &     6  \\
 \bottomrule
\end{tabular}
\end{table}

We investigate the correspondence between the detected communities and
the manually assigned categories.  We consider the largest $6$
communities of $L_0$-harvesting since they show significantly lower
conductance than the rest of communities.  The comparison is reported
in Table~\ref{table:cate-L0}.
The six largest communities \(ADC^{L_0}_1, \ldots
,ADC^{L_0}_6\) show significant correspondence to the major-fields of CS.
For example, \(ADC^{L_0}_1\)  mainly consists of two fields, operating system (OS) and
 programing (Prog), also \(ADC^{L_0}_2\) is dominated by papers from
 artificial intelligence (AI).

 Remaining smaller communities showed high precision and low recall
 with respect to the major-fields.
Some of them seem to be fragments that are not
 strongly connected to bigger communities.  However, we found that
 many of them showed correspondence to sub-fields embedded in
 major-fields.  For example, many of the small communities are related
 to AI category and they represent interactions in sub-fields of AI.

The communities detected by harvesting algorithms meet our
 expectations regarding the manual categories.  
 The detected communities revealed densely connected papers that can
 be considered as a core part within a manual category. 
 We also suspect a possible hierarchical
 community structure within the large communities and we leave
 investigations along this direction in our future work.

\subsection{Harvesting Algorithms in a Large Social Network}
\label{sec:social-network}

The massive size of modern network data, 
more than millions of nodes in a network, demands 
scalable algorithms.
Many community detection algorithms that search for optimal 
partition of nodes do not scale well.
On the other hand, harvesting algorithms detect a community 
at a time based on a locally defined quality measure. 
In this experiment, we test our harvesting
algorithms on a social network that is large and highly asymmetric.

We analyze a social network
dataset\footnote{http://www.kddcup2012.org/c/kddcup2012-track1/data}
of Tencent Weibo, a micro-blogging website of China.  Users in this
network may subscribe to news feeds from others and each
subscription is represented as a directed edge between users.  This
network contains \(\num{1944589}\) non-zero degree nodes and
\(\num{50655143}\)  edges, which leads to an average
out-degree $25$.  The social network is highly asymmetric and it has
only $0.2\%$ of symmetric links.

The computation time to harvest $\num{1000}$ $ADC^{L_{0}}$ was about
$12$ hours and that of harvesting $\num{463}$ $ADC^{EN}$ was around
$6$ hours. The algorithms are
run in a linux machine (2$\times$ Six Core Xeon X5650 / 2.66GHz /
48GB).
The settings of the algorithms can be found in Appendix~\ref{sec:algo-settings}.

To check the quality of harvested communities, we report the conductance
values, $\phi$, along with the size of $ADC$s
(Figure~\ref{fig:KDD-size-phi}). 
$L_0$-harvesting is better at detecting larger communities while
$EN$-harvesting tends to detect smaller communities and a few very
large communities.
 We also display the $\num{1000}$
largest communities obtained by \emph{Infomap}, whose directional
conductances are computed under the symmetric constraint $S=T$.
The communities found under the symmetric assumption show
relatively higher conductance values. 
Additionally, we verified that good communities are relatively small
($\sim 200$) in such huge social networks, as reported in
\cite{Leskovec:2008p2197}.

\begin{figure}[tb]
\centering
\begin{subfigure}{0.49\textwidth}
\includegraphics[width=1\textwidth]{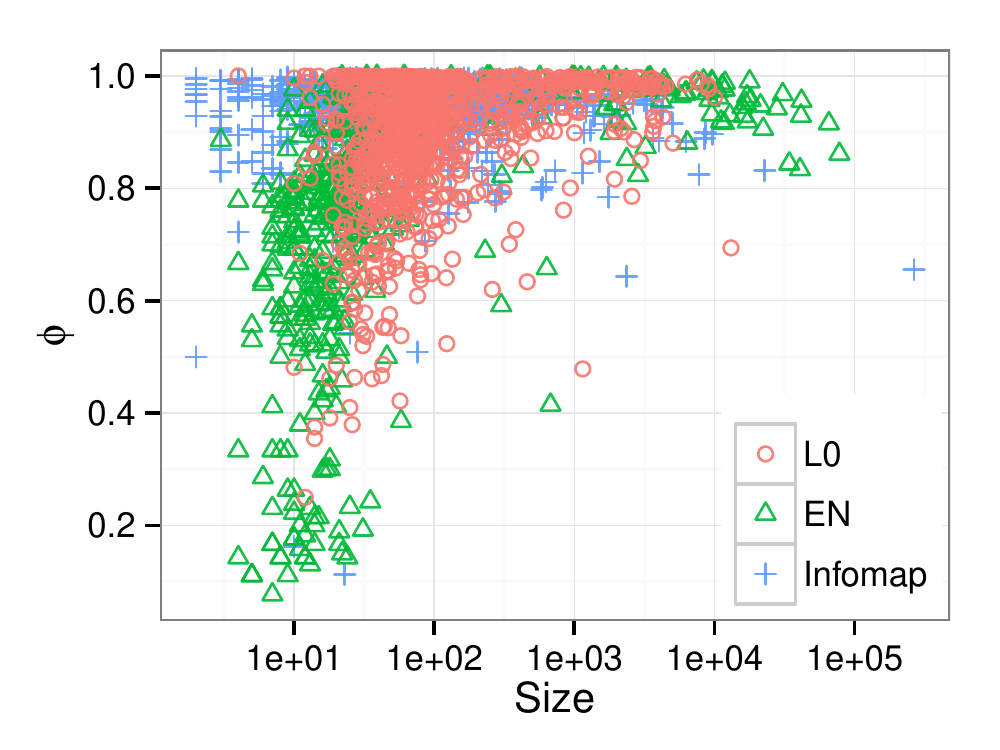}
  \caption{\label{fig:KDD-size-phi} }
\end{subfigure}
\begin{subfigure}{0.49\textwidth}
\includegraphics[width=1\textwidth]{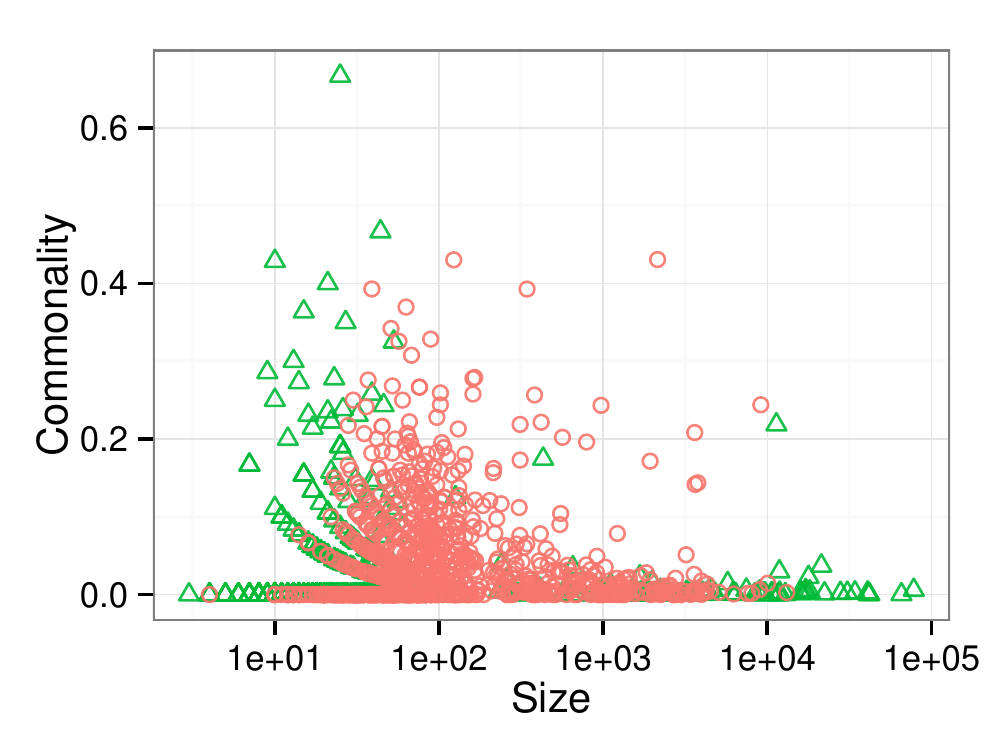}
  \caption{\label{fig:KDD-size-commonality}  }
\end{subfigure}
  \caption{\label{fig:KDD} (a) Scatter plot of  size
    of communities  and directional conductance in a social
    network. (b) Scatter plot of size of communities and commonality. }
\end{figure}

The directional communities detected by harvesting algorithms show high
asymmetricity.  We
investigate the asymmetricity of a community by looking at the ratio of
members that are common in both parts.  We define \emph{Commonality}
of a $ADC$ as the Jaccard similarity coefficient of the two parts 
(the ratio of the number of common nodes to the total
number of nodes in the union of the two parts).
Figure~\ref{fig:KDD-size-commonality} shows that most detected communities
are low in the commonality except some of small communities.  Further
inspection showed that the asymmetric communities are mostly formed by
few popular terminal nodes (authorities) and large number of source
nodes (normal users).  This observation highlights the need of
considering asymmetric directional communities in social networks.

\section{Conclusions and Discussion}
\label{sec-7}

In this paper, we found that integrating 
two different roles of nodes is critical in 
characterizing a community of a directed
network.  We introduced a new notion of community,
\emph{directional communities}, that is capable of discerning the two
different roles of each node in a community.

We proposed two regularized SVD based \emph{harvesting algorithms}
that sequentially identify directional communities. 
The regularized SVD method is linearly scalable to the
number of edges in the network.
The $L_0$-harvesting algorithm showed good
performance even in networks having low average-degree. Meanwhile the
$EN$-harvesting algorithm excelled in detecting
relatively small and dense communities.

We believe directional communities enable genuine analysis on
community structures in highly asymmetric directed networks of real
applications. Also, the simplicity of harvesting algorithms, only
relying on matrix multiplication and thresholding of vectors, leads to
further improvement of the algorithm through parallel and distributed
computing.

\section*{Acknowledgements}
This research is partially supported by NSF grant (DMS-1007060 and
DMS-130845). The authors would like to thank Dr.~Srinivasan
Parthasarathy, Dr.~Yoonkyung Lee, and Dr.~Hanson Wang for helpful discussion.

\appendix

\section{Algorithm for Finding Directional Components}
\label{app-directional-components}

This section presents a simple searching algorithm for finding directional components that is introduced in Section~\ref{sec:directional-component}.

\begin{algorithm}
 \caption{Decompose a directed graph into directional components}
\label{al:directional-components}
 \begin{algorithmic}[1]
 \Require ${\cal G}=({\cal V},{\cal E}), k = 1$
\Repeat
\State $S = \emptyset, T = \emptyset$
 \State Add a source node $s \in {\cal V}$ whose out-degree is
 non-zero in ${\cal E}$ into the set $S$. 
\Repeat
\State Find $E_t = \{e \in {\cal E} | v^s(e) \in S\}$.
\State $ T \leftarrow T \cup \{v^{t}(e) | e \in E_t\}$
\State ${\cal E} \leftarrow {\cal E} - E_t$
\State Find $E_s = \{e \in {\cal E} | v^t(e) \in T\}$
\State $ S \leftarrow S \cup \{ v^{s}(e) | e \in E_s\}$
\State ${\cal E} \leftarrow {\cal E} - E_s$
\Until{ $E_t = E_s = \emptyset$ }
\State $DC_k \leftarrow (S, T)$
\State $k \leftarrow k + 1$
\Until{${\cal E}$ is empty. }
\State \Return {$DC_1, \ldots, DC_k$}
 \end{algorithmic}
\end{algorithm}

\section{PROOFS FOR REGULARIZED SVD ALGORITHMS}
\label{app-a}

We presents the proofs of propositions of the main article.

Even though we assumed zero-one weights of edges in the main article, 
following proofs are also true for non-negative weights of edges. 
We denote the principal singular value of a matrix $X$ by $\sigma_1(X)$. 

\vskip 0.2in
\noindent  {\bf Proof of Proposition~\ref{prop:phi}}

\begin{proof}
For notational convenience, here, $\mathbf{u}(v_{i})$ is shortened to $u_i$ and $\mathbf{v}(v_{j})$ is 
shortened to $v_j$. 
We show that $ \phi_{\eta = 0}(C(S,T)) = \sum_{i,j}
W_{ij}\left(\frac{u_{i}}{\sqrt{d_{r,i}}}-\frac{v_{j}}{\sqrt{d_{c,j}}}\right)^{2}$
and at the same time
$\sum_{i,j}
W_{ij}\left(\frac{u_{i}}{\sqrt{d_{r,i}}}-\frac{v_{j}}{\sqrt{d_{c,j}}}\right)^{2}
= 1 - 2 \mathbf{u}^t Q \mathbf{v}$

\begin{align}
 \sum_{i,j}
W_{ij} \left(\frac{u_{i}}{\sqrt{d_{r,i}}}-\frac{v_{j}}{\sqrt{d_{c,j}}}\right)^{2}\notag 
  = {}&  \sum_{i \in S, j \in \bar{T}} W_{ij} \left(
 \frac{1}{\sqrt{\text{Vol}(S) + \text{Vol}(T)}}\right)^{2} + \\
{}& 
\sum_{i \in \bar{S}, j \in T} W_{ij} \left(
 \frac{1}{\sqrt{\text{Vol}(S) + \text{Vol}(T)}} \right)^{2}  \notag \\
 = {} &  \frac{\mbox{d-Cut}(C(S,T), C(\bar{S},\bar{T}))}{\text{Vol}(S) + \text{Vol}(T)} \notag
\end{align}
and on the other hand,
\begin{align}
 \sum_{i,j}
W_{ij}\left(\frac{u_{i}}{\sqrt{d_{r,i}}}-\frac{v_{j}}{\sqrt{d_{c,j}}}\right)^{2}\notag   = {}
 &  \sum_{i,j}
W_{ij}(\frac{u_{i}^{2}}{d_{r,i}}+\frac{v_{j}^{2}}{d_{c,j}} -\frac{2
  u_i v_j }{\sqrt{d_{r,i}d_{c,j}}}) \notag \\
  = {}&  \sum_i  u_i^2 + \sum_j v_j^2 -2
 \sum_{i,j} W_{ij} \frac{u_i v_j }{\sqrt{d_{r,i} d_{c,j}}}   \notag \\
 = {}& \mathbf{u}^{t} \mathbf{u} + \mathbf{v}^t \mathbf{v}-2
 \mathbf{u}^t Q \mathbf{v}  \notag \\
 = {}& 1 -2 \mathbf{u}^t Q \mathbf{v}  \notag 
\end{align}
The last equality holds by definition $ \text{Vol}(S) = \sum_{i \in S} d_{r,i} ,
\text{Vol}(T) = \sum_{j \in T} d_{c,j}$. 

\end{proof}

\vskip 0.2in
\noindent  {\bf Proof of Proposition~\ref{prop:d-comp}}
\begin{proof}
  Notice that we can modify the adjacency matrix $W$ by removing zero
  rows and zero columns without loss of generality.  The modified
  matrix is denoted by $E \in \mathbb{R}^{|{\cal S}|\times |{\cal
      T}|}$, where \({\cal S}\) is the set of source nodes whose
  out-degree is non-zero and \({\cal T}\) is the set of terminal nodes
  whose in-degree is non-zero. The singular vectors of $W$ can be
  obtained by padding zeros back to the singular vectors of $E$.

We introduce a bipartite graph expression of a directed graph that is
also considered in \cite{Zhou:2005, guimera2007module}.
The bipartite graph converted from a directed graph \({\cal G =
  (V,E)}\) is \({\cal G}_{B} = ({\cal S},{\cal T},{\cal L}) \), where
${\cal S}$ is the set of source nodes and ${\cal T}$ is the set of
terminal nodes and 
\({\cal L}\) is the set of undirected edges, $\{(v^s(e),v^t(e)), e \in
{\cal E}\}$. The adjacency matrix of ${\cal G}_{B}$, $A$, is 
\[
A =
 \left[ \begin{array}{ccc}
0 & E \\
E^t & 0 \end{array} \right].
\]

This proof has two steps, 
\begin{enumerate}
\item Show that a directional component of ${\cal G}$ is equivalent to
  a connected component of ${\cal G}_{B}$.
\item Use the relationship between the spectrum of Laplacian and connected
  components in an undirected graph to show the proposition. 
\end{enumerate}

First, let us show that a directional component (\(DC\)) in ${\cal G}$
is a connected component (\(C\)) in ${\cal G}_{B}$ by examining the
connectivity and maximality conditions:
\begin{itemize}
\item Connectivity: First, any \((s,t), s \in {\cal{S}}, t \in {\cal
    T}\) are connected in \({\cal G}_{B}\) by the D-connectivity, $s
  \leadsto t$.  Second, any \((s_1,s_2), s_1 \in {\cal{S}}, s_2 \in
  {\cal S}\) are connected in \({\cal G}_{B}\) since there exists a
  common terminal node $t \in {\cal T}$ such that $s_1 \leadsto t$ and
  $s_2 \leadsto t$.  And any \((t_1,t_2), t_1 \in {\cal T}, t_2 \in
  {\cal T}\) are connected in \({\cal G}_{B}\) for the existence of a
  common source node.
\item Maximality: Assume that there exists a node that is connected to
  \(C\) but not a member of \(DC\).  Then there should be a directed
  edge starting from the node or ended at the node in \({\cal G}\). In
  either case the node is a member of \(DC\).  It contradicts to the
  maximality of \(DC\). Thus there is no such node.
\end{itemize}
Similarly, we show that 
a connected component \(C\) in ${\cal G}_{B}$ is a directional component
\(DC\)  in \({\cal G}\).
Any pair of nodes \( (s,t), s \in {\cal S}, t \in {\cal T}
\) is D-connected in \({\cal G}\) by the connectivity in \({\cal G}_{B}\). 
Maximality for a directional component is again obtained by using the
maximality of \(C\). 

For the second step, we apply the proposition 4 of
\cite{VonLuxburg:2007p47} that shows us the equivalence between the
number of connected components of an undirected graph and the
multiplicity of the zero eigenvalue of graph Laplacian matrix of the
undirected graph.  Let \(L_{sym} \) be a normalized graph Laplacian of
$A$, which is defined by
\[
L_{sym} = I - Q_{A},
\]
where, 
\begin{align}
\label{eq:Qa}
Q_{A} & = D_{A}^{-\frac{1}{2}} A D_{A}^{-\frac{1}{2}} \notag \\
& =  \left[ \begin{array}{ccc}
0 & Q \\
Q^t & 0 \end{array} \right]
\end{align}
and \(D_A\) is the diagonal matrix of the row sums of \(A\) and it is equal to
\[
D_A =  \left[ \begin{array}{ccc}
D_r & 0 \\
0 & D_c \end{array} \right].
\]

The proposition 4 of \cite{VonLuxburg:2007p47} says that the
multiplicity \(K\) of the eigenvalue zero of \(L_{sym}\) is equal to
the number of connected components in the undirected graph
corresponding to \(A\) and the eigenspace of zero is spanned by the
vectors \(\{D_{A}^{\frac{1}{2}} \mathbf{1}_{C_k}, k = 1, \dots ,K
\}\), where \(\mathbf{1}_{C_k} \) is the indicator vector for \(k\)th
connected component.

By the definition of \(L_{sym}\), if \(\lambda\) is an eigenvalue of
\(L_{sym}\) then \(1-\lambda\) is an eigenvalue of \(Q_{A}\).  It
follows that the eigenvalue zero of \(L_{sym}\) corresponds to the
eigenvalue one of \(Q_{A}\).  In fact, one is the principal eigenvalue
of \(Q_{A}\) because the eigenvalue zero is the smallest eigenvalue of
\(L_{sym}\) which is a non-negative definite matrix.

By the standard result of the eigenvalues of \(Q_{A}\) and the
singular values of \(Q\)~\citep[see][chap. 3]{horn1994topics}, the
principal singular value of \(Q\) is the principal eigenvalue of
\(Q_A\), which is one.  A vector \(D_{A}^{\frac{1}{2}}
\mathbf{1}_{C_k} \) can be broken into two vectors
\(D_r^{\frac{1}{2}}\mathbf{1}_{S_k} \in \mathbb{R}^{|{\cal{S}}|},
D_c^{\frac{1}{2} }\mathbf{1}_{T_k} \in \mathbb{R}^{|{\cal T}|}\),
where \(D_r^{\frac{1}{2}}\mathbf{1}_{S_k} \) is the first
\(|{\cal{S}}|\) entries of \(D_{A}^{\frac{1}{2}} \mathbf{1}_{C_k} \)
and \(D_c^{\frac{1}{2}}\mathbf{1}_{T_k} \) is the last \(|{\cal T}|\)
entries of \(D_{A}^{\frac{1}{2}} \mathbf{1}_{C_k} \).  By
\eqref{eq:Qa}, the two vectors satisfy
\[
\begin{cases}
D_r^{\frac{1}{2}}\mathbf{1}_{S_k}& = Q D_c^{\frac{1}{2}} \mathbf{1}_{T_k} \\
D_c^{\frac{1}{2}}\mathbf{1}_{T_k} & = Q^{t} D_r^{\frac{1}{2}}\mathbf{1}_{S_k},
\end{cases}
\]
as one can find in \cite{Dhillon:2001}.  \(
\{D_r^{\frac{1}{2}}\mathbf{1}_{S_k}, k = 1, \dots, K\} \) is a set of
orthogonal vectors since \(S_k\)'s are exclusive.  The same argument
holds for \( \{D_c^{\frac{1}{2}}\mathbf{1}_{T_k}, k = 1, \dots, K\}
\).  Thus, the pairs of vectors \(
\{(D_r^{\frac{1}{2}}\mathbf{1}_{S_k},
D_c^{\frac{1}{2}}\mathbf{1}_{T_k}), k = 1, \dots, K\} \) span the
singular space of the singular value one of \(Q\).
 
\end{proof}

Using the adjacency matrix expression of a directed graph, a 
directional component  can be considered as a submatrix of a matrix. 
For a non-negative matrix $B$,
we call a submatrix of  $B$ a directional-component block if the
submatrix is corresponding to a directional component of the directed
graph generated from the weight matrix $B$. 

We introduce a corollary of Proposition~\ref{prop:d-comp}. 
This corollary is used in the proof of Theorem~\ref{thm:connectivity} later. 

\begin{corollary}
  \label{cor:submatrix}
  For any submatrix of $Q$, say $Q_s $, the largest singular value of
  $Q_{s}$ is less than or equal to one ( $\sigma_1(Q_s)\leq 1 $), 
 and the equality holds if and only if $Q_s$
  includes directional-component blocks. 
\end{corollary}

\begin{proof}
First of all, we introduce a handy representation of a submatrix $Q_s \in \mathbb{R}^{k \times l}$.
A submatrix of $Q$ is a matrix formed by selecting a subset of rows
and columns of $Q$. 
We define a full-rank matrix, called a selection matrix, whose columns
have only one non-zero entry
with its value. 
Then, for any submatrix $Q_s$, we can find two selection matrices $M_r \in \mathbb{R}^{m \times
  k}, M_c \in \mathbb{R}^{n \times l}$ such that
\[
Q_s = M_r^t Q M_c,
\]
according to the selected rows and columns.

The principal singular value of $Q_s$, $\sigma_1(Q_s)$,
is the solution of a optimization problem,
\begin{equation}
  \label{eq:lem-11-1}
  \max_{\mathbf{u}_s, \mathbf{v}_s } \mathbf{u}_s^t Q_s \mathbf{v}_s, \quad \|\mathbf{u}_s\|_2 = 1,  \|\mathbf{v}_s\|_2 = 1.   
\end{equation}
with $\mathbf{u}_s \in \mathbb{R}^k, \mathbf{v}_s \in \mathbb{R}^l $. 
By setting $\mathbf{u} = M_r\mathbf{u}_s, \mathbf{v} = M_c\mathbf{v}_s$, 
we can see that \eqref{eq:lem-11-1} is equivalent to 
\begin{equation}
  \label{eq:lem-11-2}
  \max_{\mathbf{u}_s, \mathbf{v}_s } \mathbf{u}^t Q \mathbf{v}, \quad \|\mathbf{u}\|_2 = 1,  \|\mathbf{v}\|_2 = 1, \mathbf{u} = M_r\mathbf{u}_s, \mathbf{v} = M_c\mathbf{v}_s
\end{equation}
by $\|M_r\mathbf{u}_s\|_2 = \|\mathbf{u}_s\|_2, \|M_c\mathbf{v}_s\|_2 = \|\mathbf{v}_s\|_2 $. 
This optimization has constraints, $\mathbf{u} = M_r\mathbf{u}_s, \mathbf{v} = M_c\mathbf{v}_s$, in addition
to the formulation of the principal singular value of $Q$. 
Thus, $\sigma_1(Q_s) \leq 1$
by Proposition~\ref{prop:d-comp}.

Proposition~\ref{prop:d-comp} also tells us that $\sigma_1(Q_s) = 1$ if and 
only if $(u,v) \in \chi_1$ at the solution of
\eqref{eq:lem-11-2}, where $\chi_1 \subset \mathbb{R}^{n+m}$ is the principal singular space of $Q$.
Thus, it is clear that $\sigma_1(Q_s) = 1$ if and 
only if $\chi_1 \cap \chi_{(M)} \neq \mathbf{0}$,
 where,
\[\chi_{(M)} =
 \mbox{span}\{\{(M_{r,i},\mathbf{0}_m)\}_{i=1,\ldots,k} \cup
 \{(\mathbf{0}_n,M_{c,i})\}_{i=1,\ldots,l} \},
\]
where $M_{r,i}$ is the $i$-th column vector of $M_r$. 

Therefore it is enough to show that $\chi_1 \cap \chi_{(M)} \neq \mathbf{0}$
 if and only if $Q_s$ includes directional component
blocks. 
We want to clarify that this statement is about the condition on
$M_r, M_c$, which is equivalent to the condition on the selected rows
and columns of $Q$ for $Q_s$.

We start to show one direction by taking an non-zero vector $(u,v) \in \chi_1
\cap \chi_{M}$.
Since $(u,v) \in \chi_1$, $(u,v)$ should have non-zero entries at the
same places of non-zero entries of 
$(\mathbf{1_{S_k}}, \mathbf{1_{T_k}})$ for some $k$.
$(u,v)$ also belongs to $\chi_{M}$, thus the span of the columns of
$M_r$ have to
include $\mathbf{1_{S_k}}$ and also the span of the columns of $M_c$
have to include $\mathbf{1_{T_k}}$. 
Therefore, we conclude that $Q_s$ includes $(S_k, T_k)$ and it is true
for any $k$. 

The other direction can be shown easily by setting $Q_s$ to include a \(k\)th directional
component block of $Q$. Then, $(D_r^{\frac{1}{2}}\mathbf{1}_{S_k},
D_c^{\frac{1}{2}}\mathbf{1}_{T_k}) \in \chi_1 \cap \chi_M$. 

\end{proof}

Now we show the solution of an optimization problem,
  \begin{equation}
      \label{eq:ideal}
      \max_{\mathbf{u} ,\mathbf{v} } \, \mathbf{u}^tQ \mathbf{v}  -\eta(\|\mathbf{u}\|_0 +
      \omega \|\mathbf{v}\|_0), \qquad \|\mathbf{u} \|_2 \leq 1,  \quad \|\mathbf{v}\|_2 \leq 1,
   \end{equation}
provides a D-connected directional community.

\vskip 0.2in
\noindent  {\bf Proof of Theorem~\ref{thm:connectivity}}

\begin{proof}
Given membership vectors $\mathbf{u}, \mathbf{v}$ and the corresponding community $C(S,T)$, notice that
$\|\mathbf{u}\|_{0} = |S|$ and $\|\mathbf{v}\|_{0} = |T|$. 
We obtain a matrix $Q(C(S,T))$ by setting the rows and columns of $Q$
that are not in $S,T$ to zero vectors. 
Then, \eqref{eq:ideal} can be written as
\begin{equation}
  \label{eq:L0-form}
\max_{S,T} \sigma_1(Q(C(S,T))) -\eta SZ_{\omega}(C(S,T))  
\end{equation}
Suppose a solution $C(S^{*}, T^{*})$ of \eqref{eq:L0-form} is not
D-connected and can be decomposed into several maximal D-connected
communities within $C(S^{*}, T^{*})$.  Then $\sigma_1(Q(C(S^{*},
T^{*})))$ is equal to the principal singular value of one of the
D-connected communities. But the size of the D-connected community is
smaller than the size of $C(S^{*}, T^{*})$. Thus 
the objective function of \eqref{eq:L0-form} can be increased by the
smaller D-connected community. 
This contradicts the
supposition that $C(S^{*}, T^{*})$ maximizes the objective function. 

Since a directional component is maximal D-connected subgraph, any
D-connected subgraph should be a subgraph of some directional
component.

We prove the second claim. 
Corollary~\ref{cor:submatrix} tells us that $\sigma_1(Q(DC_1))$  is equal
to $1$ and 
that is one of the largest among $\{\sigma_1(Q(C(S,T))) |
SZ_{\omega}(C(S,T)) \geq SZ_{\omega}(DC_1) \}$. 
Thus all $C(S,T)$ such that $SZ_{\omega}(C(S,T)) >
SZ_{\omega}(DC_1)$ can not be a solution. 
We consider the condition of $\eta$ that satisfies 
\[
1 - \eta SZ_{\omega}(DC_1)  < \sigma_{1}(Q(C(S,T))) - \eta SZ_{\omega}(C(S,T)),
\]
for some $C(S,T)$ such that $SZ_{\omega}(C(S,T)) < SZ_{\omega}(DC_1)$. 
After an arrangement of above inequality, $\eta$ should satisfy
\begin{equation}
  \label{eq:eta-bound}
\eta > \frac{1-\sigma_{1}(Q(C(S,T)))}{SZ_{\omega}(DC_1) - SZ_{\omega}(C(S,T))}.  
\end{equation}
$SZ_{\omega}(DC_1) - SZ_{\omega}(C(S,T)) > 0$ by the condition of $C(S,T)$ and
$1-\sigma_{1}(Q(C(S,T))) > 0$ by Corollary~\ref{cor:submatrix},
 thus taking minimum over the possible communities finishes the
 proof. 
\end{proof}

We provide a toy example illustrating how the regularized SVD can detect the smallest directional component. 
The left panel of Figure~\ref{fig:singular-submatrix-1} 
shows the adjacency matrix of a network with ten nodes. 
The network has two directional components with 
different sizes. 
 The parameter $\omega$ in  
 defining the size of communities is  set to 
$1.1$.
We randomly select two subsets $(S, T)$ 
of nodes to generate a submatrix $Q(C(S,T))$ from the full graph
Laplacian matrix $Q$.
Consider $S,T$ as
indexes of rows and columns of $Q$ respectively.  
For each selected $Q(C(S,T))$, we 
calculate its principal singular value $\sigma_1(Q(C(S,T)))$ 
and its size $SZ_{\omega}(C(S,T))$. 
In addition to 500 randomly chosen sub-matrices,
those two directional components are included as references. 

\begin{figure}[h]
\begin{subfigure}{1\textwidth}
  \centering
  \includegraphics[trim = 0px 200px 0px 100px, clip,width=.9\linewidth]{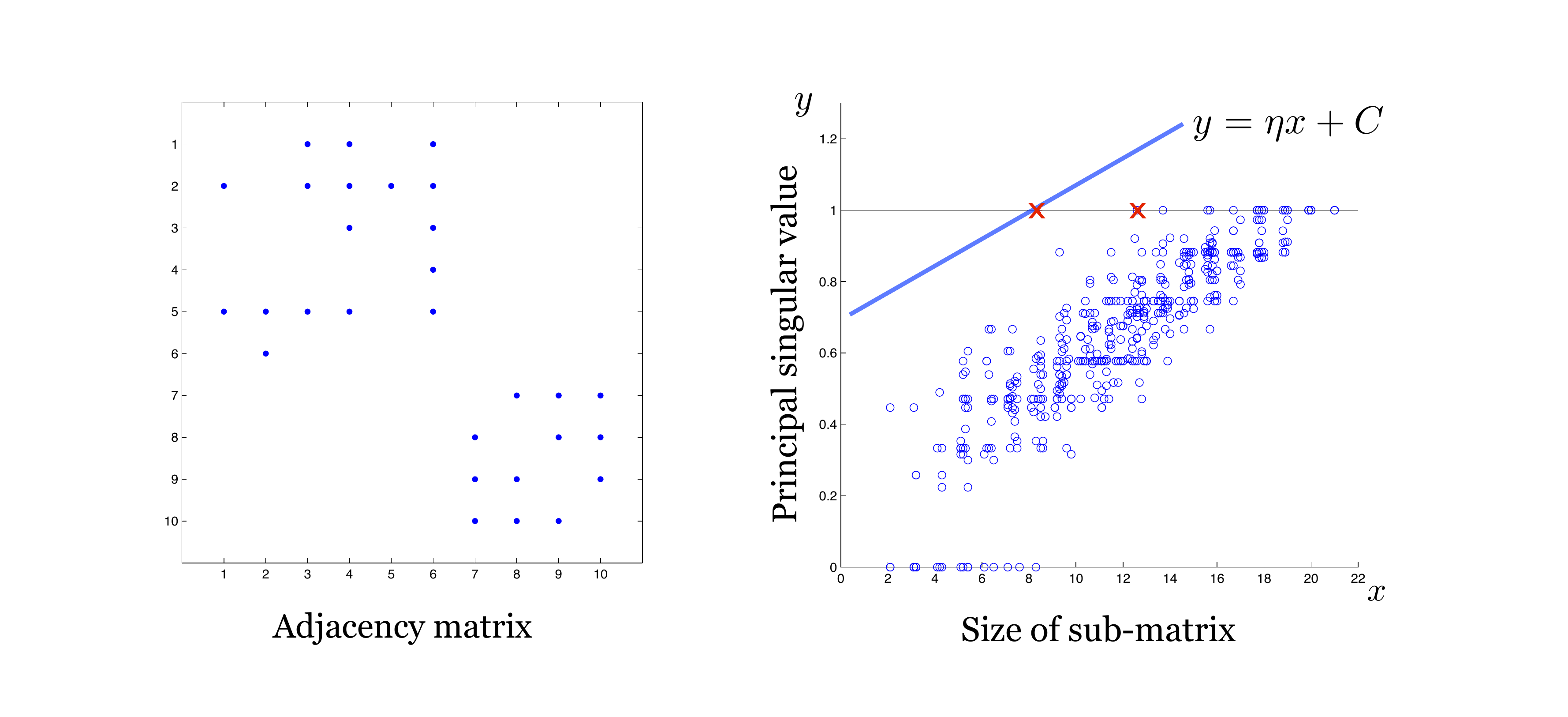}
  \caption{\label{fig:singular-submatrix-1} No external edges}
\end{subfigure}

\begin{subfigure}{1\textwidth}
  \centering
  \includegraphics[trim = 0px 200px 0px 100px, clip,width=.9\linewidth]{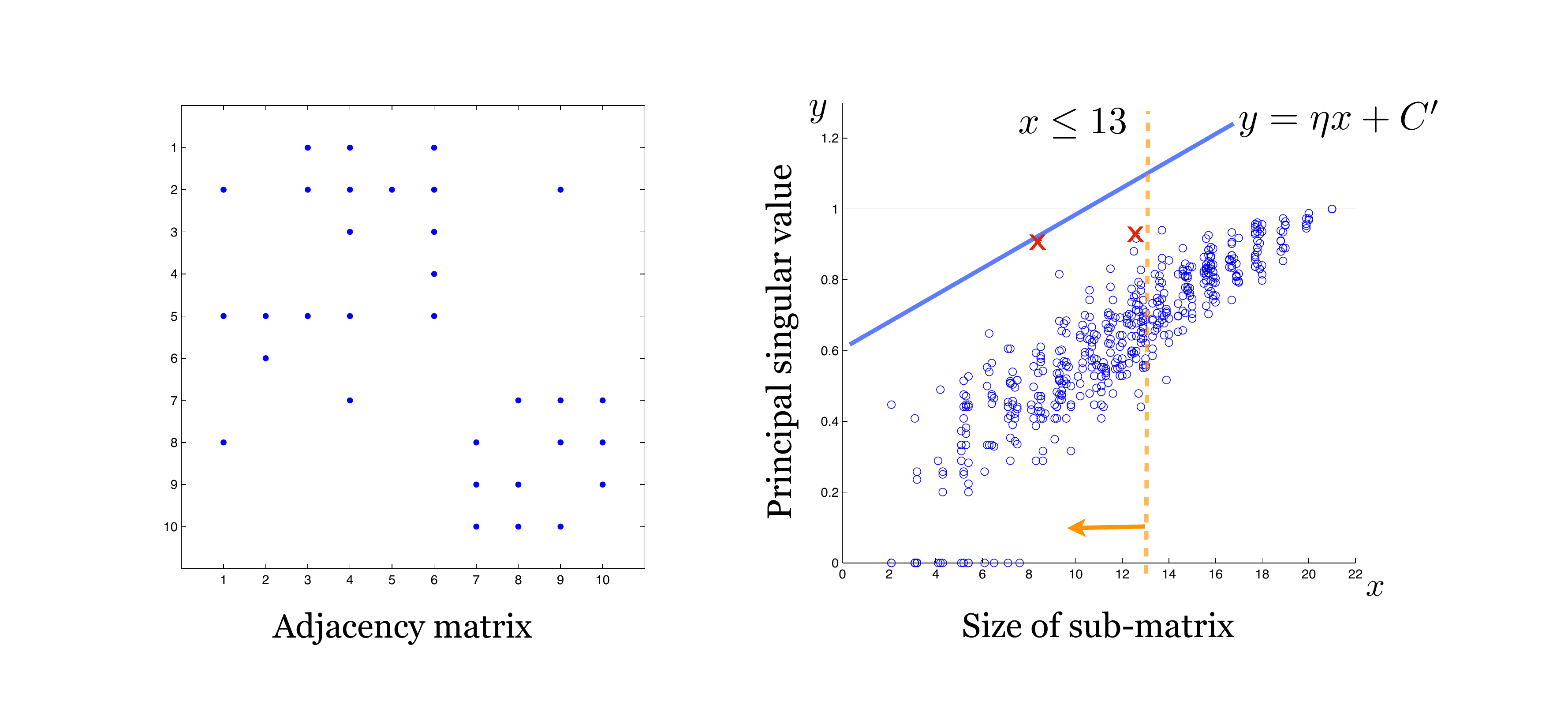}
  \caption{\label{fig:singular-submatrix-2} After adding three external edges}
\end{subfigure}
\caption{\label{fig:singular-submatrix} Left panel of (a): The adjacency
    matrix of an example network having two directional
    components. Right panel of (a):  Scatter plots of 
    $SZ_{\omega}(C(S,T))$ and $\sigma_1(Q(C(S,T)))$. $Q(C(S,T))$ is a submatrix  of the 
    graph Laplacian matrix $Q$ derived from a directed graph having two directional
    components. 
Left panel of (b): The adjacency
    matrix of the example network of
    Figure~\ref{fig:singular-submatrix-1} after adding three external
    edges. Right panel of (b):  Scatter plots of 
    $SZ_{\omega}(C(S,T))$ and $\sigma_1(Q(C(S,T)))$. $Q(C(S,T))$ is a
    submatrix  of the 
    graph Laplacian matrix $Q$ derived from the directed graph corrupted by
    external edges.}
\end{figure}

The right panel of Figure~\ref{fig:singular-submatrix-1} 
show paired values $(SZ_{\omega}(C(S,T)), \sigma_1(Q(C(S,T))))$, with {\color{blue}{o}}, 
for each sampled submatrix $Q(C(S,T))$, and  those two directional components 
are marked as {\color{red}{X}}s.
Let us denote the value of objective function in \eqref{eq:ideal} as
$C$. This figure shows that there exist a line with slope $\eta > 0$ whose
intercept $C$ is maximized at the point corresponding to the smallest
directional component as  Theorem~{3.3} describes. 
Therefore, the optimization \eqref{eq:ideal} leads to identification of the  
smaller of the two directional components in this network.
To summarize the result, both Proposition~\ref{prop:d-comp} 
and the example show that directional components, if there is any, 
can be identified sequentially by $L_0$ regularized SVD approach.

Recall that we encountered the problem that the small number of external edges 
connect small directional communities together as one large directional component. 
The root of the problem is the too strict requirement on finding exact 
directional components, the maximal set of node satisfying D-connectivity.  
The $L_0$ regularized SVD limits the number of non-zero entries of the singular vectors, 
so it may find a community that is smaller and
almost separated from other communities.

To illustrate the advantage of the regularized approach, 
we add three external edges in the example. 
As a result, those two directional components merge together as one,
as shown in the left  panel of Figure~\ref{fig:singular-submatrix-2}.
The right panel of Figure~\ref{fig:singular-submatrix-2} 
plots paired values $(SZ_{\omega}(S, T), \sigma_1(Q(C(S,T))))$
of the same 500 pairs of $(S,T)$
shown in Figure~\ref{fig:singular-submatrix-1}.
The principal singular values of two true directional components 
({\color{red} X} marks) have decreased because of the added external edges, 
but the line with the same slope is still capable of identifying the
original directional component.
In addition to the argument of approximately minimizing penalized directional conductance, this example shows that the regularized SVD may capture smaller communities that are embedded in a directional component. 

\vskip 0.2in
\noindent  {\bf Proof of Proposition~\ref{prop:optim-L0}}

\begin{proof}
For a fixed number of non-zero entries $\|\mathbf{u}\|_0 = l$, $\max_{\|\mathbf{u}\|_2 \leq 1} \mathbf{u}^t\mathbf{z} $ is obtained when 
$
\mathbf{u} = {\mathbf{z}^h_l}/ {\|\mathbf{z}^h_l\|_2}
$.
Thus the objective function (\ref{def:maxL0}) can be written as 
\[
F(l) = \|\mathbf{z}^h_l\|_2- \rho \, l.
\]
Now we maximize $F(l)$ over $l$. 
Notice that $ \|\mathbf{z}^h_l\|_2$ increases monotonically as $l$ increases. The value of
$F(l)$ keeps increasing until 
\[
\sqrt{\|\mathbf{z}^h_l\|_2^2 + |z|_{(l+1)}^2} - \|\mathbf{z}^h_l\|_2 \leq \rho, 
\]
which is equivalent to \eqref{eq:optim-L0}.
After $l$ goes beyond this point, 
 $F(l) $ starts to decrease
and keeps decreasing because $|z|^2_{(l)}$ decreases and $\|\mathbf{z}^h_l\|_2$ increases.
Therefore, the solution to \eqref{def:maxL0} is obtained at the minimum $l$ that satisfies \eqref{eq:optim-L0}. 
\end{proof}

We begin to show an optimization problem, 
\begin{equation}
\label{eq:EN-u}
\max_{\mathbf{u} } \mathbf{u}^{t} \mathbf{z}, \quad
  \mbox{subject to } \quad (1-\alpha) \|\mathbf{u} \|^2_{2} + \alpha \|\mathbf{u} \|_{1} \leq c_1, 
\end{equation}
can be solved by soft-thresholding. First, we introduce a definition.

\begin{definition}{3.5.}
For a vector $\mathbf{z}=(z_1, \ldots, z_n)' \in \mathbb{R}^n$, recall the $l$-th largest 
absolute value of $\mathbf{z}$ was defined as ${|z|}_{(l)}$.
We define
\begin{align}
\label{eq:23}
G_{\mathbf{z}}(x) &=  \frac{1}{4x^2} \sum_{i = 1}^{k(x)-1} (|z|_{(i)}- x)^2+ 
\frac{1}{2x} \sum_{i = 1}^{k(x)-1} (|z|_{(i)} - x) 
\end{align}
where $ k(x) \in \{ 1, \dots , n+1 \} $ satisfies $ |z|_{(k(x))} \leq
x < |z|_{(k(x)-1)} $ and define $|z|_{(n+1)} = 0$. 
\end{definition}

\vskip 0.2in
\noindent  {\bf Proof of Proposition~\ref{thm:elastic}}

\begin{proof}
The first part of this proof
resembles the proof of Lemma 2.2 of \cite{Witten:2009}. 
Express the objective function and the constraints by using a Lagrangian multiplier,
\begin{align}
\label{eq:20}
\min_{\mathbf{u},\lambda}  & -\mathbf{u}^{t}\mathbf{z} + \lambda ( (1-\alpha) \|\mathbf{u}\|_{2}^{2} + \alpha \|\mathbf{u}\|_{1}).
\end{align}
Then, differentiate the objective function in \eqref{eq:20} by
$\mathbf{u}$ and set it to zero, 
\begin{align}
 -\mathbf{z} + \lambda(2(1-\alpha) \mathbf{u} + \alpha \Gamma) = 0, \notag
\end{align}
where \(\Gamma_i = \mbox{sign}(u_i)\) if \(u_i \not= 0\), otherwise 
\(\Gamma_i \in [-1,1] \). The Karush-Kuhn-Tucker (KKT) conditions require \(\lambda ( (1-\alpha) \|\mathbf{u}\|_{2}^{2} + \alpha \|\mathbf{u}\|_{1}-c_1) = 0 \).
If \(\lambda > 0\),  the solution is 
\begin{align}
\hat{\mathbf{u}} = \frac{S(\mathbf{z}, \lambda \alpha)}{2\lambda (1-\alpha)}. \notag
\end{align}
\(\lambda\) can be zero, if the solution is not on the boundary of the constraint. But it does not happen unless \(\mathbf{z}\) is a zero vector. 
 Thus, \(\lambda > 0\) is chosen so that \(\hat{\mathbf{u}}\) satisfies the KKT condition. 
\begin{align}
\label{eq:21}
&( 1-\alpha ) \left\|\frac{S(\mathbf{z}, \lambda \alpha)}{2\lambda ( 1-\alpha )} \right\|_2^2 + \alpha \left\| \frac{S(\mathbf{z}, \lambda \alpha)}{2\lambda ( 1-\alpha )} \right\|_1 = c_1 \notag \\
\Rightarrow & \frac{1}{(2 \lambda)^2 ( 1-\alpha )} \sum_{i = 1}^{k-1} (|{z}|_{(i)}- \lambda \alpha)^2 + \frac{\alpha}{2 \lambda ( 1-\alpha )} \sum_{i = 1}^{k-1} (|{z}|_{(i)} - \lambda \alpha) = c_1 
\end{align}
where \(k\) satisfies \(|{z}|_{(k)} \leq \lambda\alpha <  |{z}|_{(k-1)} \). 
Denote the threshold level \( d = \lambda \alpha \), then \eqref{eq:21} becomes
\begin{align}
\label{eq:22}
 \left( \frac{1}{4d^2} \sum_{i = 1}^{k-1} (|{z}|_{(i)}- d)^2 + \frac{1}{2d} \sum_{i = 1}^{k-1} (|{z}|_{(i)} - d)  \right) = c_1 \frac{1-\alpha}{\alpha^2}, 
\end{align}
where \(k\) satisfies \(|{z}|_{(k)} \leq d <  |{z}|_{(k-1)} \).
Using Lemma~\ref{lem:thresh}, one can determine the threshold level \(d\) of \eqref{eq:22} by setting \( \mathbf{z} \) and \(c = c_1 \frac{1-\alpha}{\alpha^2} \). 
Even though the value of \(\lambda\) is not required for the solution, 
we present it for the record. 
\begin{align}
\lambda = \frac{1}{\alpha} \left(\frac{\sum_{i = 1}^{\hat{k}} |{z}|_{(i)}^2}{4(c_1\frac{1-\alpha}{\alpha^2}) + \hat{k}} \right)^{\frac{1}{2}}. \notag
\end{align}
\end{proof}

Now, we show how to obtain the threshold level in Proposition~{3.6}.

\begin{lemma}
\label{lem:thresh}
The solution of the equation $G_{\mathbf{z}}(d) = c$ for given $c > 0$  is 
\begin{align}\label{eq:k}
d = \left(\frac{\sum_{i = 1}^{\hat{k}} |z|_{(i)}^2}{4c + \hat{k}}\right)^{\frac{1}{2}},
\end{align}
where $ \hat{k} $ is a positive integer in $\{1,2, \dots, n \}$ that satisfies 
$G_{\mathbf{z}}(|z|_{(\hat{k})}) \leq c,  G_{\mathbf{z}}(|z|_{(\hat{k}+1)}) > c$. 
\end{lemma}

\begin{proof}
  For the first step, we show that \( G_{\mathbf{z}}(\cdot)\) is a
  monotone decreasing function, that is, if \( d_1 > d_2 \), then \(
  G_{\mathbf{z}}(d_1) < G_{\mathbf{z}}(d_2) \).  The first term of
  \eqref{eq:23} is monotone decreasing of \(d\) because
\begin{align}
\label{eq:24}
 \frac{1}{4d_2^2} \sum_{i = 1}^{k(d_2)-1} (|z|_{(i)}- d_2)^2 & >  \frac{1}{4d_1^2} \sum_{i = 1}^{k(d_1)-1} (|z|_{(i)}- d_2)^2 \notag \\
& >  \frac{1}{4d_1^2} \sum_{i = 1}^{k(d_1)-1} (|z|_{(i)}- d_1)^2. \notag
\end{align}
The first inequality comes from the fact that \( k(d_2) \leq k(d_1) \)
and \( d_1 > d_2\). 
The second inequality comes from \(d_1 > d_2 \). 
The second term of \eqref{eq:23} can be done in the similar way and the desired result is obtained. 

For the second step, we find the approximated solution of \(d\) from the set of \( \{ |z|_{(i)}\}_{i = 1 \ldots n}\). 
By plugging in \( |z|_{(i)} \) to \(d\) in the increasing order of
\(i\), we can find \( \hat{k} \) such that
\(G_{\mathbf{z}}(|z|_{(\hat{k})}) \leq c \), \( G_{\mathbf{z}}(|z|_{(\hat{k}+1)}) > c \) by the monotonicity of \( G_{\mathbf{z}}(\cdot) \) and being \(c\) in the range of \(G_{\mathbf{z}}(\cdot)\). 
This computation can be done efficiently by computing two cumulative
sums, \(\sum_{i}^{k}|z|_{(i)}^2 \) and \( \sum_{i}^{k}|z|_{(i)} \), in the increasing order of $k$ until $\hat{k}$ is
obtained. 
An algorithm for finding $\hat{k}$ in this Lemma is provided 
  in Algorithm~\ref{al:hatk}.

\begin{algorithm}
\caption{Find $\hat{k}$ such that
\(G_{\mathbf{z}}(|z|_{(\hat{k})}) \leq c \), \( G_{\mathbf{z}}(|z|_{(\hat{k}+1)}) > c \)}
\label{al:hatk}
\begin{algorithmic}[1]
\Require {$ (z_1 \geq z_2 \geq, \dots, \geq z_n), c > 0 $}
\State initialize $S_1 \gets 0, S_2 \gets 0, \hat{k} \gets 2 $
\For{$k = 2:n$}
\State $ S_1 \gets S_1 + z_{k-1} $
\State $ S_2 \gets S_2 + z_{k-1}^{2} $
\State $ G_k = \frac{1}{4z_{k}^2} (S_2 - 2 z_k S_1 + (k-1)z_{k}^2) + \frac{1}{2 z_k}(S_1 - (k-1)z_k) $
\If {$ G_k > c $}
\State $ \hat{k} \gets k-1$
\State \Return  {$\hat{k}$}
\EndIf
\EndFor
\If {$ G_k \leq c $}
\State $\hat{k} \gets n $
\State \Return  {$\hat{k}$}
\EndIf
\end{algorithmic}
\end{algorithm}

In the second step, we already know that \(|z|_{(\hat{k}+1)} < d \leq |z|_{(\hat{k})}\) which means \(k=\hat{k}\) fixed now.  Therefore solving a quadratic equation of \(d\), 
\begin{align}
 \frac{1}{4d^2} \sum_{i = 1}^{\hat{k}} (|z|_{(i)}- d)^2+ \frac{1}{2d} \sum_{i = 1}^{\hat{k}} (|z|_{(i)} - d)  = c \notag
\end{align}
determines the solution \(d\). 
By the quadratic formula, the solution is
\begin{align}
d = \left(\frac{\sum_{i = 1}^{\hat{k}} |z|_{(i)}^2}{4c + \hat{k}}\right)^{\frac{1}{2}}, \notag
\end{align}
knowing that $d > 0$. 
\end{proof}

\section{Algorithm Settings for Simulation}
\label{app-simulation-settings}

The true number of
communities $N_C$ is provided for DI-SIM algorithm. 
We compute the first $N_C$ singular vectors of \(Q\) and apply
the k-means algorithm with $N_C$ clusters on the left and right
singular vectors separately. 
We run k-means algorithm with $100$ random initializations 
and the one minimizing the within-cluster sums of
point-to-cluster-centroid distances is reported as the final result.
To obtain directional communities out of the two separate partitions,
we match the source part and the terminal part by the largest
common edges. 

We use Infomap implementation \texttt{Infomap-0.11.5} available at
\url{http://www.mapequation.org}. The default settings are used except for
the options for directed links (\texttt{--directed}) and two-level partition of the
network (\texttt{--two-level}). 
100 repetitions (\texttt{--num-trials}) are used to pick the best
solution.

Harvesting algorithms are initialized with \(v_0\) being the node of
largest in-degree at each harvesting.
The sparsity levels for source part and terminal part are set to the
same value, $\omega = 1$ 
and $\alpha = \beta$.
The range of them are determined so that the detected communities
are sized roughly $SZ_{\omega = 1}(C) \in (20, 400)$,  
More specifically, the grid of sparsity levels for $L_0$ penalty, $\eta$, contains \(10\)
points in \( \{ \exp(-k): k = 6+i (5/10), i = 1, \ldots, 10 \} \) and 
the grid of sparsity levels for $EN$ penalty, $\alpha$, includes \(10\) points in \( \{\frac{1}{1+\exp(k)}: k = 1 + i (3.7/10), i = 1, \ldots, 10 \} \). 
Those non-linear grids are adapted for more constant change of the
size of candidate communities. 
Early stopping parameters are set to \(s_p= 1.5\) and \( s_l=0.6\).
The harvesting algorithm continues until the number of harvested
communities reaches the true number of communities or there is no more
edge left.

\section{Communities in Cora Citation Network}
\label{app-c}

Here we present details of $ADC$s found in Cora Citation Network. 
Table~\ref{table:summary-d-comps} shows four summaries of $20$
$ADC^{L_0}$ and $ADC^{EN}$ ordered by the size. 
 \(|S|\) and \(|T|\)  are the number of source nodes and terminal
 nodes,  \(|E|\) is the number of edges in the community. 
 $\phi$  is the value of directional conductance. 

Also we provide a comparison between manual categories of papers and
communities detected by four algorithms ($L_0$-harvesting:
Table~\ref{table:cate-L0-2}, $EN$-harvesting:
Table~\ref{table:cate-EN}, DI-SIM: Table~\ref{table:cate-DI-s})
Those tables show the number of papers in manually assigned categories
for each community. 

\begin{table}[h]
\caption{\label{table:summary-d-comps}Summary of the largest \(20\) $ADC$s of
  Cora citation network. \(|S|\) is the number of source nodes and
  \(|T|\) is the number of terminal nodes and  \(|E|\) is the number
  of edges.  $\phi$ stands for directional conductance. \\}
        \centering
    \begin{subtable}{0.45\textwidth}
\begin{tabular}{c|rrrr}
  \toprule
  Order & \(|S|\) & \(|T|\)  & \(|E|\) & $\phi$  \\
  \midrule
  1   &   3266  &  2321  &  21851  &  0.1500  \\        
  2   &   2636  &  1886  &  12972  &  0.2244  \\       
  3   &   1543  &  1128  &   8342  &  0.1724  \\       
  4   &   1381  &   971  &   4690  &  0.2034  \\       
  5   &   1270  &   919  &   6037  &  0.1910  \\       
  6   &    803  &   512  &   3790  &  0.1271  \\       
  7   &    694  &   480  &   4143  &  0.3638  \\       
  8   &    577  &   485  &   2299  &  0.4906  \\       
  9   &    573  &   447  &   2018  &  0.3070  \\       
  10  &    583  &   361  &   2455  &  0.4363  \\       
  11  &    539  &   368  &   2522  &  0.3033  \\       
  12  &    503  &   403  &   1580  &  0.3588  \\       
  13  &    587  &   278  &   1750  &  0.4666  \\       
  14 &     479  &   251  &   1659  &  0.2909  \\       
  15 &     390  &   278  &   1558  &  0.3031  \\       
  16 &     368  &   233  &    938  &  0.4609  \\       
  17 &     370  &   207  &   1007  &  0.3271  \\       
  18 &     334  &   171  &    970  &  0.2416  \\       
  19 &     291  &   207  &   1119  &  0.2312  \\       
  20 &     226  &   154  &    672  &  0.4978   \\    
  \bottomrule
\end{tabular}
\caption{\label{table:summary-d-comps-L0} First \(20\) $ADC^{L_0}$. }
    \end{subtable}
\quad
    \begin{subtable}{0.45\textwidth}
      \centering
      \begin{tabular}{c|rrrr}
\toprule
      Order & \(|S|\) & \(|T|\)  & \(|E|\) & $\phi$  \\
     \midrule
     1   &   5319  &  3176  &  25428  &  0.2579  \\     
     2   &   4458  &  2756  &  17137  &  0.2437  \\     
     3   &   2309  &  1535  &  10422  &  0.2626  \\     
     4   &   2254  &  1546  &  14539  &  0.2176  \\     
     5   &    914  &   650  &   3127  &  0.3839  \\     
     6   &    752  &   488  &   3219  &  0.3605  \\     
     7   &    643  &   444  &   2522  &  0.4176  \\     
     8   &    528  &   323  &   1561  &  0.3223  \\     
     9   &    441  &   304  &   1487  &  0.3702  \\     
     10  &    453  &   276  &   1602  &  0.2505  \\     
     11  &    258  &   139  &   1504  &  0.2965  \\     
     12  &    225  &   164  &    987  &  0.3794  \\     
     13  &    245  &   116  &   1515  &  0.2070  \\     
      14 &    195  &   130  &    558  &  0.3265  \\     
      15 &    187  &   136  &    555  &  0.5642  \\     
      16 &    187  &   132  &    629  &  0.2128  \\     
      17 &    191  &   120  &    512  &  0.3706  \\     
      18 &    162  &    94  &    512  &  0.2834  \\     
      19 &    141  &   115  &    510  &  0.4501  \\     
      20 &    168  &    80  &    430  &  0.2624  \\   
 \bottomrule
\end{tabular}
\caption{\label{table:summary-d-comps-EN} First \(20\) $ADC^{EN}$. }
\end{subtable}
\end{table}

\begin{table}[ht]
  \centering
\caption{\label{table:cate-L0-2}Number of papers in the first twenty 
  approximated directional components of \(L_0\)-harvesting for each category. }
\begin{tabular}{l|rrrrrrrrrrr}
  \toprule
     &    AI  &  DSAT  &   DB  &   EC  &   HA  &  HCI  &  IR  &  Net  &   OS  &  Prog   \\
  \midrule
1    &       106  &   199  &   56  &   18  &  255  &   17  &   0  &   55  &  900 \cellcolor[gray]{0.9} &  1779 \cellcolor[gray]{0.9}  \\
2     &      2741 \cellcolor[gray]{0.9}  &    68  &   30  &    9  &    8  &   28  &  63  &   17  &   34  &    75   \\
3     &      13  &    12  &   25  &  115  &   11  &  307  &  18  &  936 \cellcolor[gray]{0.9}  &  232  &    34 \\
4     &      727  \cellcolor[gray]{0.9} &   124  &    8  &  102  &   12  &  577  \cellcolor[gray]{0.9} &  10  &    6  &   22  &    21  \\
5     &      284  &    83  &  803 \cellcolor[gray]{0.9}  &    9  &    9  &   17  &  66  &   14  &   16  &    80 \\
6     &      149  &   452  \cellcolor[gray]{0.9}&    3  &  239 \cellcolor[gray]{0.9} &   12  &    0  &   2  &    3  &    9  &     6  \\
7     &      16  &    40  &   95  &   19  &   94  &   32  &   7  &   50  &  347  &    96 \\
8     &      40  &    90  &   14  &   14  &   32  &   11  &   0  &  112  &  254  &   157 \\
9     &       283  &   184  &    0  &    8  &   29  &   19  &   1  &    3  &   32  &    37 \\
10   &       18  &    38  &   30  &   13  &    2  &   28  &   0  &   27  &  355  &   127  \\
11   &        651  &     1  &    0  &    2  &    1  &    1  &  24  &    0  &    0  &     4  \\
12   &       524  &     7  &    3  &    1  &    0  &   22  &  73  &    1  &    2  &    22 \\
13   &       543  &    23  &    1  &    3  &    2  &    1  &  45  &    0  &    9  &    31 \\
14   &        492  &     4  &   10  &    8  &    8  &    2  &   1  &    0  &    0  &     3 \\
15   &        427  &    11  &    0  &    8  &    0  &    1  &   3  &    0  &    3  &     3 \\
16   &        104  &     6  &   23  &    3  &    3  &  187  &  12  &    0  &   13  &   110   \\
17   &        21  &     9  &    3  &  307  &    3  &    8  &   2  &   20  &   40  &    22 \\
18   &        20  &    66  &    2  &    0  &  221  &    0  &   0  &   26  &    6  &    15  \\
19   &        243  &    14  &    0  &    7  &   15  &    1  &   0  &    1  &   12  &    34  \\
20   &        292  &     1  &    1  &    3  &    0  &    5  &   7  &    0  &    0  &     0   \\
 \bottomrule
\end{tabular}
\end{table}

\begin{table}[ht]
  \centering
\caption{\label{table:cate-EN}Number of papers in the first twenty 
  approximated directional components of \(EN\)-harvesting for each category.}
\begin{tabular}{l|rrrrrrrrrrr}
\toprule
         & AI   & DSAT & DB  & EC  & HA  & HCI  & IR  & Net & OS   & Prog \\
\midrule
     1  &   573  &  341  &  649  &  184  &  175  &  396  &   83  &  449  &  1136  &  1817  \\                 
     2  &  4015  &   78  &   90  &   49  &   37  &  134  &  365  &   24  &    96  &   167  \\                 
     3  & 2218  &   39  &   56  &   25  &    9  &   43  &   49  &   88  &    61  &    92  \\                 
     4  &   80  &  238  &   42  &   11  &  214  &   17  &    1  &   71  &   891  &   809  \\                 
     5 &    214  &  583  &   62  &   21  &   42  &   62  &    2  &   70  &    32  &    44  \\                 
     6 &     25  &   18  &    5  &   78  &    6  &   97  &    9  &  587  &    62  &    15 \\                 
     7  &   708  &   12  &   11  &   13  &   13  &    5  &    3  &    0  &     2  &     8  \\                 
     8  &  186  &  186  &   11  &   12  &   31  &    6  &    0  &    6  &    23  &    52  \\                 
     9  &    75  &   37  &  103  &    0  &    7  &    0  &    0  &    0  &     6  &   293 \\                 
     10 &    15  &   31  &    0  &  394  &    4  &    2  &    2  &   28  &    38  &     6 \\                 
     11 &     0  &    1  &    2  &   18  &    0  &   38  &    0  &  192  &    23  &     0  \\                 
     12 &    26  &    8  &    1  &    2  &    0  &    0  &    0  &    0  &     1  &   220   \\                 
     13 &    76  &  132  &    1  &   37  &    1  &    0  &    0  &    0  &     4  &     2  \\                 
      14 &   95  &   30  &    0  &    0  &    0  &  120  &    0  &    1  &     1  &     4  \\                 
      15 &   16  &   14  &  173  &    7  &    4  &    5  &    6  &    2  &    12  &    15  \\                 
      16 &  169  &   33  &    0  &    0  &    0  &    8  &    0  &    0  &     2  &     2  \\                 
      17 &   99  &    9  &    0  &    6  &    0  &  126  &    0  &    0  &     1  &     0   \\                 
      18 &    5  &    9  &    1  &   11  &  135  &    1  &    1  &    3  &     0  &    15 \\                 
      19 &    0  &    1  &    2  &    1  &    4  &    0  &    0  &    8  &    98  &    49  \\                 
      20 &   13  &    3  &    0  &    0  &  136  &    0  &    0  &    3  &     3  &     5  \\                 
 \bottomrule
\end{tabular}
\end{table}

\begin{table}[ht]
  \centering
\caption{\label{table:cate-DI-s} Number of papers in the source
  partition DI-SIM for each category.}
\begin{tabular}{l|rrrrrrrrrrr}
\toprule
         & AI   & DSAT & DB   & EC  & HA  & HCI & IR  & Net  & OS   & Prog \\
\midrule
     1   &   687  &  1084  &  723  &  575  &  571  &  416  &   71  &  750  &  1176  &  1933  \\
     2   &    28  &     4  &    0  &    0  &    0  &    0  &    0  &    0  &     1  &     0 \\
     3   &     4  &     0  &    0  &    0  &    0  &    0  &    0  &    0  &     0  &     0  \\
     4   &  2650  &    14  &   18  &   21  &    6  &   47  &   95  &    1  &     5  &    14 \\
     5   &   120  &   165  &   78  &   85  &  177  &  173  &   13  &  489  &  1023  &  1075 \\
     6   &   100  &    42  &    5  &   14  &   13  &   17  &    5  &   10  &    18  &    23  \\
     7   &  2509  &  1374  &  269  &  316  &  373  &  506  &  126  &  288  &   305  &   779 \\
     8   &    13  &     8  &    0  &   18  &    0  &    3  &    0  &    0  &     0  &     0   \\
     9   &  4658  &   406  &  167  &  149  &   64  &  485  &  272  &   20  &    50  &   148  \\
     10  &    15  &     7  &    1  &    3  &    3  &    4  &    0  &    3  &     2  &     0 \\
 \bottomrule
\end{tabular}
\end{table}

\section{HARVESTING ALGORITHM SETTINGS FOR SOCIAL NETWORK DATA}
\label{sec:algo-settings}

The sparsity level parameters in the
harvesting algorithms are designed to capture communities with the
size in the range of $10$ to $\num{100000}$ approximately.  The grid
of sparsity parameter \(\eta\) in \(L_0\)-harvesting is set as \( \{
\exp(-k): k = 10+i (13/50), i = 1, \ldots, 50 \} \) and the grid for
$\alpha =\beta$ in \(EN\)-harvesting is set as \(
\{\frac{1}{1+\exp(k)}; k = 5 +i (6/50), i = 1, \ldots, 50 \} \).  The
early stopping method is applied with the parameters \(s_p= 1.1\) and
\( s_l=0.8\).

\clearpage
\vskip 0.2in
\bibliography{paper_comu_detec_2012}

\end{document}